\documentclass[11pt]{article}

\usepackage{amssymb}

\usepackage[ruled,linesnumbered,vlined,noend]{algorithm2e}

\newcommand{\eam}{{\sc Edge Arrival}\xspace}
\newcommand{\ea}{\mbox{{\sc Ea}}\xspace}

\newcommand{\vam}{{\sc Vertex Arrival}\xspace}
\newcommand{\va}{\mbox{{\sc Va}}\xspace}
\newcommand{\alm}{{\sc Adjacency List}\xspace}
\newcommand{\al}{\mbox{{\sc Al}}\xspace}

\newcommand{\ra}{\mbox{{\sc Ra}}\xspace}
\newcommand{\rouone}{{\sc SampLightHelper}\xspace}
\newcommand{\routwo}{{\sc SampHeavyHelper}\xspace}

\newcommand{\tsamp}{{\sc Triangle Sampling}\xspace}
\newcommand{\poly}{\mbox{{poly}}}

\newcommand{\defproblem}[3]{
  \vspace{5mm}
\noindent\fbox{
  \begin{minipage}{0.96\textwidth}
  \begin{tabular*}{\textwidth}{@{\extracolsep{\fill}}lr} #1 \\ \end{tabular*}
  {\bf{Input:}} #2  \\
  {\bf{Output:}} #3
  \end{minipage}
  }
  \vspace{5mm}
}

\newcommand{\sampllight}{{\sc Sample-Light-Triangle}\xspace}

\newcommand{\samplheavy}{{\sc Sample-heavy-Triangle}\xspace}

\usepackage{todonotes}

\newcounter{mynotes}
\renewcommand{\hat}{\widehat}

\usepackage[margin = 1.0in]{geometry}

\usepackage{microtype}%if unwanted, comment out or use option "draft"
\usepackage{epsfig}
\usepackage{amssymb}
\usepackage{microtype}
\usepackage{wrapfig}
\usepackage{graphicx}
\usepackage{epsfig,amsmath,amsfonts,amssymb,amsthm,mathrsfs,ifpdf}
\usepackage{indentfirst,relsize}
\usepackage[numbers]{natbib}
\usepackage{setspace}
\usepackage{enumerate}
\usepackage{latexsym}
\usepackage{stackrel}
\usepackage[all]{xy}
\usepackage[usenames,dvipsnames]{pstricks}
\usepackage{pst-grad} % For gradients
\usepackage{pst-plot} % For axes
\usepackage{xspace}
\usepackage{bbm}

\usepackage[super]{nth}
\usepackage{hyperref}

\usepackage[T1]{fontenc}
\usepackage{multirow}

\usepackage{lineno}

\usepackage[noend]{algpseudocode}
\usepackage{cleveref}

\newcommand{\size}[1]{\left| #1 \right|}

\newcommand{\E}{\mathbb{E}}
\newcommand{\remove}[1]{}

\newcommand{\N}{\mathbb{N}}

\newcommand{\cM}{\mathcal{M}}

\newcommand{\cT}{\mathcal{T}}

\newcommand{\cD}{\mathcal{D}}
\newcommand{\cU}{\mathcal{U}}
\newcommand{\cE}{\mathcal{E}}
\newcommand{\cF}{\mathcal{F}}

\newcommand{\Oh}{\mathcal{O}}

\newcommand{\tOh}{\widetilde{{\Oh}}}

\newcommand{\eps}{\varepsilon}
\newcommand{\pr}{\mathbb{P}}

\newcommand{\comments}[1]{\textcolor{blue}{\bf{#1}}}

\theoremstyle{plain}
\newtheorem{theo}{Theorem}[section]
\newtheorem{lem}[theo]{Lemma}

\newtheorem{coro}[theo]{Corollary}

\theoremstyle{definition}
\newtheorem{defi}[theo]{Definition}
\newtheorem{rem}{Remark}
\newtheorem{obs}[theo]{Observation}
\newtheorem{qs}[theo]{Question}

\title{
{\bf Near Uniform Triangle Sampling Over Adjacency List\\
Graph Streams}
}
\author{
Arijit Bishnu \footnote{Indian Statistical Institute, Kolkata, India. Email: arijit@isical.ac.in.}
\and
Arijit Ghosh \footnote{Indian Statistical Institute, Kolkata, India. Email: arijitiitkgpster@gmail.com.}
\and
Gopinath Mishra
\footnote{National University of Singapore, Singapore. Email: gopianjan117@gmail.com.}
\and
Sayantan Sen\footnote{Centre for Quantum Technologies,
National University of Singapore, Singapore. Email: sayantan789@gmail.com.}
}

\date{}

\begin{document}
%\begin{titlepage}
%\clearpage\thispagestyle{empty}

\maketitle

\begin{abstract}
Triangle counting and sampling are two fundamental problems for {\em streaming algorithms}.  Arguably, designing sampling algorithms is more challenging than their counting variants. It may be noted that triangle counting has received far greater attention in the literature than the sampling variant. In this work, we consider the problem of \emph{approximately sampling triangles} in different models of streaming with the focus being on the \emph{adjacency list} model.

In this problem, the edges of a graph $G$ will arrive over a data stream. The goal is to design efficient streaming algorithms that can sample and output a triangle from a distribution, over the triangles in $G$, that is \emph{close} to the uniform distribution over the triangles in $G$. The distance between distributions is measured in terms of $\ell_1$-distance. The main technical contribution of this paper is to design algorithms for this triangle sampling problem in the \emph{adjacency list} model with the space complexities matching their counting variants.  For the sake of completeness, we also show results on the vertex and edge arrival models. 
\end{abstract}

\thispagestyle{empty}

%\clearpage
%\pagenumbering{arabic} 
%\end{titlepage}

\section{Introduction}
\label{sec:intro}
\noindent
The primary goal of this paper is to initiate a systematic study of \emph{sampling almost uniformly} a triangle over a data stream, denoted as $\cM$. Our setting is the usual streaming models (namely Adjacency List (AL) model, Edge Arrival (\ea) model and Vertex Arrival (\va) model) with single or multi-passes~\cite{muthukrishnan2005data,mcgregor2014graph} over a graph $G=(V, E)$ with the set of vertices $V$ and the set of edges $E$ such that $\size{V}=n$  and $\size{E}=m$. Let us start with a brief description of the streaming models.
\begin{itemize}
 \item {\alm(\al):} The set of vertices $V$ of $G$ are revealed in an arbitrary order. When a vertex $v \in V$ is revealed, all the edges that are incident to $v$, arrive one by one in an arbitrary order. Note that in this model each edge arrives twice, once for each endpoint.

\item \eam(\ea): The stream consists of the edges of $G$ in an arbitrary order.

\item \vam(\va): The set of vertices $V$ of $G$ are revealed in an arbitrary order. When a vertex $v \in V$ is revealed, all the edges between $v$ and the neighbors of $v$ that are already revealed, arrive one by one in an arbitrary order.

\end{itemize}

The contribution of this paper is in designing triangle sampling algorithms in all the three models with the primary focus being on the \al model {as we can show that the existing triangle counting algorithms in the \ea and \va models can be extended to triangle sampling}. 
For a graph $G$ having at least one triangle, let $t(G)$ denote the set of triangles present in $G$, and $\mathcal{U}$ denote the uniform distribution over the set $t(G)$. 
The formal description of the problem that we study in this work is as follows:

\defproblem{\underline{{\sc Triangle Sampling} over Data Streams}}{A graph $G(V,E)$ with its edges as a stream over $p$  pass(es), and a proximity parameter $\eps \in (0,2)$.}{A triangle sampled almost uniformly, that is, from a  distribution $\mathcal{D}$ over $t(G)$ such that the $\ell_1$-distance between  $\mathcal{D}$ and $\mathcal{U}$ is at most $\eps$.~\footnote{For two distributions $\cD_1$ and $\cD_2$ over a (non-empty) set $S$, the $\ell_1$-distance between $\cD_1$ and $\cD_2$ is denoted and defined as 
$
    |\size{\cD_1 -\cD_2}|_{1} := \sum_{x \in S}\size{\cD_1(x)-\cD_2(x)}.
$}}

In this work, we determine the space complexities of \tsamp in \al, \va, and \ea models that match the space complexity of the counting variants of the problems.

\medskip 

In terms of applications, triangle counting as well as triangle sampling have found wide applications across multiple disciplines such as computational biology~\cite{scott2006efficient,shlomi2006qpath}, databases~\cite{durand2013structural,chen2017logic,arenas2021approximate}, social network analysis~\cite{ugander2013subgraph}, to name a few. In particular, sampling over triangle queries and general join queries as well as sampling over join-project queries are studied in the database community (see \cite{chen2020random,zhao2018random} and the references therein). Moreover, this is also used in the analysis of huge graphs as well as for community detection in large networks~\cite{bloedorn2005relational,milo2002network}.

\medskip

\paragraph{The complexity of sampling may be more than that of counting.}
It may be pertinent to mention here that though counting triangles over a data stream has received wide attention for a long time now~\cite{bera2017towards,kallaugher2019complexity,braverman2013hard, buriol2006counting,cormode2014second,bulteau2016triangle,jowhari2005new, kallaugher2017hybrid,kolountzakis2012efficient,DBLP:conf/pods/McGregorVV16,pagh2012colorful,pavan2013counting}, to the best of our knowledge, the only work prior to ours that dealt with sampling triangle over a data stream (in the $1$-pass \ea/\va model) is the work by Pavan et al.~\cite{pavan2013counting}. 
 Though triangle counting has been a challenging as well as an important problem in the streaming literature, it seems that designing algorithms for  triangle sampling is possibly harder than its counting variant. To count triangles {or any other implicit structures, like cycles, cliques, etc.,} that are aggregating measures, one works on random samples of different types. But there is no guarantee that those random samples would serve the purpose of figuring out a uniformly random sampled triangle or such an implicit structure. Here, it may be instructive to see what happened in the ``other'' sublinear world, i.e. in the property testing models pertaining to query complexity.

 It is interesting to note that in the local query model, where we have only \emph{degree}, \emph{edge existence} and \emph{neighborhood} query access to a graph,~\footnote{In degree query, given a vertex $v$, the oracle returns $deg(v)$. 
For edge-existence query, given two vertices $u$ and $v$, the oracle returns $1$ if there is an edge between $u$ and $v$. Finally, in the neighborhood query, given a vertex $v$ and an integer $i$, the oracle returns the $i$-th neighbor of $v$ if it exists, or a special symbol if $deg(v) < i$.} the query complexity of sampling triangles is more compared to its counting counterparts~\cite{eden2022almost}. However, the addition of \emph{random edge} query access to the graph makes the query complexities of triangle counting and sampling the same ignoring polynomial terms in $1/\eps$ and $\log n$~\cite{DBLP:conf/icalp/FichtenbergerG020}. So, it is interesting to understand the status of \tsamp in the streaming setting in the following sense: are the space complexities of \tsamp the same as their counting counterparts ignoring polynomial terms in $1/\eps$ and $\log n$? In this work, we answer this question in the affirmative for \al, \ea and \va models.

\subsection{Our Results}
\label{ssec:results}
\noindent
The main technical contribution of the paper is designing optimal algorithms for {\sc Triangle Sampling} in the \al model.
%\complain{(Gopi: why there are two labels in the statement of the theorem?)}
\begin{theo}{\bf (Main Result: Informal Statement).}
\label{theo:two}
\tsamp can be solved in \al model using: (i) $\widetilde{\Theta}\left(m/T^{2/3}\right)$ space in $3$-passes and (ii) $\widetilde{\Theta}\left(m/\sqrt{T}\right)$ space in $1$-pass.
    Here $T$ is a promised lower bound on the number of triangles in $G$.~\footnote{The promise $T$ on the number of structures  is a standard assumption for estimating substructures (such as triangles) in the world of graph streaming algorithms~\cite{DBLP:conf/focs/KallaugherKP18,DBLP:conf/icalp/KaneMSS12,kallaugher2019complexity, DBLP:conf/pods/McGregorVV16, bera2017towards}. Here we have cited a few. However, there are huge amount of relevant literature.} Moreover, $\widetilde{\Theta}(\cdot)$ hides a polynomial in $\log n$ and $1/\eps$ in the upper bound.
\end{theo}
Let us discuss about the complexity of \tsamp in \ea/\va model. It is known that \tsamp can be solved in \ea/\va model in $1$-pass by using $\tOh\left(\min\{m,m^2/T\}\right)$ space~\cite{pavan2013counting}.~\footnote{In Appendix~\ref{sec:eaalgo}, we also present a new algorithm of \tsamp in $1$-pass \ea/\va model with space complexity $\widetilde{\Theta}\left(\min\{m,m^{2}/T\}\right)$ which we believe has a unifying structure with our other algorithms in different settings.} We observe that multi-pass triangle counting algorithms in \ea/\va model~\cite{assadi2018simple} can be extended directly to \remove{get algorithms for \tsamp, that is, we can} solve \tsamp in \ea/\va  model using $\widetilde{\Theta}\left(m^{3/2}/T\right)$ space in $3$-passes.
The details of the algorithms in \ea/\va models are discussed in Appendix~\ref{sec:eaalgo}. Our algorithm pertaining to Theorem~\ref{theo:two} is inspired from the counting algorithms in \al model \cite{kallaugher2019complexity, DBLP:conf/pods/McGregorVV16}. However, our sampling algorithms undergo a significantly different and subtle analysis compared to their counting counterparts. We believe that the detailed analysis constitutes the primary contribution of our paper.
\remove{
\comments{In our paper, we acknowledge that the foundational inspiration is drawn from the respective counting algorithms for \tsamp in the \al model \cite{kallaugher2019complexity, DBLP:conf/pods/McGregorVV16}. However, it is important to highlight that our sampling algorithms undergo a significantly different and certail analysis compared to their counting counterparts. We believe that the detailed analysis constitutes the primary contribution of our paper.}}

\medskip

The proof of the two upper bounds in Theorem~\ref{theo:two}, presented in \Cref{sec:almultipassalgo} and \Cref{sec:al1passalgo}, respectively form the main technical component of this paper. The lower bounds in Theorem~\ref{theo:two} and the lower bounds in \ea/\va model
follow from the  lower bounds for counting version of the problems~\cite{kallaugher2019complexity,braverman2013hard,bera2017towards}. The details are discussed in \Cref{app:rem} and \Cref{sec:lb_app}.

% \begin{rem}

%     

\remove{
\paragraph{Challenges of designing sampling algorithms.} In the local query model, the sampling variant of problems like \emph{edge sampling} and \emph{triangle sampling} took some time to be solved after the complexities of their counting variants have been resolved. In particular, the query complexity of edge estimation is resolved by Goldreich and Ron ~\cite{DBLP:journals/rsa/GoldreichR08} and that of edge sampling is resolved by Eden and Rosenbaum~\cite{eden2018sampling}. Also, the query complexity of triangle estimation is resolved by Eden, Levi, Ron, and Seshadri~\cite{eden2017approximately} and that of triangle sampling is resolved by Eden, Ron, and Rosenbaum~\cite{eden2018sampling}. This discussion clearly indicates that resolving the space complexity of triangle sampling is an important problem given that the space complexities of the counting variants have been resolved and the sampling variants have not been explored in detail.
}
%In this work, we answer the above question in negative. In three popular streaming models (\al, \va and \ea models), we design triangle sampling algorithms which have the similar space complexity as their counting counterparts, ignoring  polynomial terms in $1/\eps$ and $\log n$.

\paragraph{Notations.}
As in other works in streaming, we use a parameter $T$ to denote a guaranteed lower bound on the total number of triangles in $G$, i.e., $\size{t(G)} \geq T$. To make the exposition simple while describing our algorithms, we assume that we know the exact values of $m$ and $\size{t(G)}=T$. However, the assumptions can be removed by using standard techniques in graph streaming. The details are discussed in \Cref{app:rem}. \remove{We will use $e=\{u,v\}$ as an usual (unordered) edge, whereas, $e= (u,v)$ will denote its ordered variant. }\remove{For an edge $e\in E$, the number of triangles with $e$ as an edge is denoted as $t_e$. An edge $e \in E$ is said to be $\tau$-heavy, if $t_e \geq \tau$, for some integer $\tau \in \N$.} 
We will use $\widetilde{\Oh}(\cdot)$ to hide polynomial factors in $1/\eps$ and $\log n$. For an event $\cE$ and $\alpha, p \in (0,1)$, we write $\pr(\cE)=(1\pm \alpha)p$ to denote that the probability of the event $\cE$ lies between $(1-\alpha)p$ and $(1+\alpha)p$. When we say that event $\cE$ holds with high probability, we mean $\pr(\cE)\geq 1-1/\mathsf{poly}(n)$. Let $S$ be a non-empty set and $\cD$ be a probability distribution over $S$. For $x \in S$, $\cD(x)$ denotes the probability of sampling $x$ when one draws a random sample from $\cD$. For two distributions $\cD_1$ and $\cD_2$ over a (non-empty) set $S$, the $\ell_1$-distance between $\cD_1$ and $\cD_2$ is denoted and defined as 
$
    |\size{\cD_1 -\cD_2}|_{1} := \sum_{x \in S}\size{\cD_1(x)-\cD_2(x)}.
$
{\emph{Reservoir sampling} is a widely used randomized algorithmic framework in streaming that samples $k$ uniformly random elements from a stream where $k$ is known before the start of the stream~\cite{DBLP:journals/toms/Vitter85}. The space complexity of the algorithm is $\Oh(k)$. We denote a reservoir sampler by RS and explicitly mention whether we are considering sampling with or without replacement.
}
In a similar fashion, one can consider \emph{weighted reservoir sampler}, denoted by WRS, to sample elements proportional to their weights.

\subsection{Related Works}
\label{ssec:relatedwork}
\noindent
Alon et al.~\cite{alon1997finding} studied the problem of counting triangles in directed as well as undirected graphs in the RAM model. \remove{and proved that deciding whether an undirected (directed) graph has a triangle or not and finding a triangle if one exists can be achieved in $\Oh(|E|^{2\omega/\omega +1})$ time, where $E$ denotes the set of edges of $G$ and $\omega \leq 2.372$ denotes the matrix multiplication exponent.} The problem of triangle counting in the streaming model was first considered by Bar-Yossef et al.~\cite{bar2002reductions} where the authors studied the problem of approximately counting triangles in a graph $G$ where the edges of $G$ appear as a stream. Subsequently, there have been several interesting works in several variants~\cite{braverman2013hard,buriol2006counting,cormode2014second,bulteau2016triangle,jowhari2005new,kallaugher2017hybrid,kolountzakis2012efficient,DBLP:conf/pods/McGregorVV16,pagh2012colorful,pavan2013counting,kallaugher2019complexity,JayaramK21}.
For an exhaustive list of results, and several techniques, see the surveys of Muthukrishnan~\cite{muthukrishnan2005data} and McGregor~\cite{mcgregor2014graph}; also see Bera and Chakrabarti~\cite{bera2017towards} for a nice summary of relevant results.

A close cousin of counting is the problem of sampling. However, in the field of streaming algorithms, sampling implicit graph structures like triangle, clique, and cycles have not been very well studied before. In the field of property testing, where sublinear number of query accesses to the graph is important,  uniform sampling of edges has been considered~\cite{eden2018sampling, DBLP:conf/stoc/TetekT22,DBLP:conf/sosa/EdenNT23}. Later there have been results on sampling cliques \cite{eden2022almost} and sampling arbitrary subgraphs \cite{DBLP:conf/icalp/FichtenbergerG020,biswas2021towards}.

As seen from the preceding discussions, although triangle counting has been studied extensively in the streaming model, sampling triangles uniformly from a graph stream has not been well studied before, excepting for a work or two. Sampling works in property testing setting is also of recent vintage. In this work, we progress towards closing this gap between the study of counting and sampling triangles over a graph stream.

Now, we discuss briefly relevant triangle counting literature. For single pass \ea/\va model, Braverman et al.~\cite{braverman2013hard} proved that even for distinguishing triangle-free graphs from graphs with $T$ triangles, $\Omega(m)$ space is required. Pavan et al.~\cite{pavan2013counting} studied the problem of triangle counting and sampling in $1$-pass \ea/\va model, where they designed an algorithm that uses $\widetilde{\Theta}(m^2/T)$ space. Note that the work of Pavan et al.~\cite{pavan2013counting} is the only work that studied the triangle sampling problem prior to this work. Their triangle sampling technique first picks up an edge $e$ uniformly at random using reservoir sampling. Next they pick up an edge uniformly at random from the adjacent edges of $e$. This creates a wedge; they now wait for the arrival of an edge that completes the triangle with the edge $e$.
Cormode and Jowhari~\cite{cormode2014second} studied the triangle counting problem in \ea model and designed a multi-pass algorithm with space complexity $\widetilde{\Oh}(m/\sqrt{T})$. Later McGregor et al.~\cite{DBLP:conf/pods/McGregorVV16} designed two separate multi-pass algorithms with space complexity $\tOh(m^{3/2}/T)$ and $\tOh(m/\sqrt{T})$, respectively. Bera and Chakrabarti~\cite{bera2017towards} also proved a lower bound of $\Omega(\min\{m^{3/2}/T, m/\sqrt{T}\})$ for multi-pass setting, and presented a simpler multi-pass algorithm with space complexity $\widetilde{\Oh}(m^{3/2}/T)$. \remove{Below we give a high-level overview of that algorithm. } For \al model, McGregor et al.~\cite{DBLP:conf/pods/McGregorVV16} designed a single pass $(1 \pm \eps)$-multiplicative approximation algorithm for counting the number of triangles in  $\widetilde{\Oh}(m/\sqrt{T})$ space. \remove{They classified all edges into two categories: \emph{heavy} and \emph{light} edges, where an edge $e=(a, b)$ is said to be heavy if $t_e \geq \sqrt{T}$, and light otherwise, where $t_e$ denotes the number of triangles incident to the edge $e$. They define an estimator which has a low variance for light edges. Since the number of triangles associated with heavy edges is maintained among the sampled edges, they conclude that the number of triangles can be approximated by using $\widetilde{\Oh}(m/\sqrt{T})$ space.} The multi-pass algorithm for triangle counting in \al model was first studied by the same authors in \cite{DBLP:conf/pods/McGregorVV16} which was later improved by Kallaugher et al.~\cite{kallaugher2019complexity} to $\widetilde{\Oh}(m/T^{2/3})$.

\section{Overview of our algorithms}
\label{sec:overview}

\noindent In this section, we present a technical overview of our upper bound results of \Cref{theo:two}. The corresponding lower bounds follow from a sampling to counting reduction technique, and are discussed in~\Cref{sec:lb_app}.

\subsection{Multi-pass triangle sampling algorithm in the \al model}\label{sec:al-multi-over}
\noindent Triangle is an implicit structure derived from edges; so to get to an uniformly random sampled triangle, we target the triangles via edges in the following way. We call an edge \emph{heavy} (\emph{light}) if it has \emph{many} (\emph{less number of}) triangles incident on it. Informally speaking, an edge is said to be heavy if the number of triangles  incident on it is at least some threshold $\tau=\Theta_{\varepsilon}(T^{1/3})$.~\footnote{$\Theta_{\eps}(\cdot)$ hides the dependency on $1/\eps$.} Otherwise, the edge is said to be light. A triangle is said to be heavy if all of the edges of the triangle are heavy. Otherwise, the triangle is said to be light. One can argue that the total number of heavy edges is at most $\Theta_{\varepsilon}(T^{2/3})$ and the number of heavy triangles is $\Oh(\varepsilon T)$. So, if we can show that we can sample light triangles uniformly, we are done. To do so, consider the following algorithm: take an edge $e$ uniformly at random. If $e$ is light, then choose a triangle $s_e$ uniformly at random from the triangles incident on $e$. Then compute the number of edges of $s_e$ that are light, let it be $i\in \{1,2,3\}$ ({$i \not= 0$}  as $e$ itself is a light edge). Report $s_e$ as the sampled triangle with probability $\frac{1}{i}\cdot \frac{\lambda_e}{\tau}$, where $\lambda_e$ denotes the number of triangles incident on $e$. Now, consider a light triangle $\Delta$ such that the number of light edges on it is $j \in \{1,2,3\}$. The probability that $\Delta$ is sampled is $\frac{j}{m}\cdot \frac{1}{\lambda_e}\cdot \frac{1}{j}\cdot \frac{\lambda_e}{\tau} =\frac{1}{m\tau}=\Theta_{\eps}(\frac{1}{mT^{1/3}})$. Since the number of heavy triangles is $\Oh(\eps T)$, the number of light triangles is $\Omega(T)$. Thus the probability that the above procedure reports some light triangle is $\Theta_{\eps}(T^{2/3}/m)$. If the above process succeeds, we can argue that we get a light triangle uniformly at random. To boost the success probability, one can repeat the above procedure $\tOh{(m/T^{2/3})}$ times.

Here, we would like to emphasize that the above intuition can be formalized and implemented in streaming model by exploiting the properties of the \al model. We defer the details to \Cref{sec:almultipassalgo}.

Note that the intuition described above can be worked out to give an alternative multi-pass algorithm for triangle counting in the \al model with the same space bound as of \cite{kallaugher2019complexity}. The intuition {behind the triangle counting algorithm} of \cite{kallaugher2019complexity} is different from the intuition we described above for our sampling algorithm. It is not clear to us if the algorithm of \cite{kallaugher2019complexity} can be extended to sampling.

\subsection{Single pass triangle sampling algorithm in the \al model} \label{sec:al-onepass-over}

\noindent We will use an idea akin to light and heavy edges as described in the 3-pass algorithm, but to get to a 1-pass algorithm, we will take recourse to running subroutines in parallel. The idea of the algorithm is to charge each triangle to a unique edge depending on the order of exposure of vertices in the stream. An edge is designated to be either \emph{heavy} or \emph{light}~\footnote{Here, the definition of heavy and light are different than that of the 3-pass algorithm. The exact definitions will be given in the respective section.}, depending on the number of triangles charged to it. A triangle is said to be heavy or light according to the type of the edge to which it has been charged. Over the stream, our algorithm detects all the heavy edges. For each heavy edge, our algorithm samples a triangle uniformly at random from all the triangles incident on that heavy edge. Our algorithm also samples a light triangle with a \emph{suitable} probability. Let us discuss the intuition in some more detail. 

For two vertices $a,b \in V$, we say $a \prec_s b$ {(the $s$ in $\prec_s$ denotes the order imposed by the streaming order of the vertices)}  if the adjacency list of a is revealed before the adjacency list of $b$. Consider a triangle $\Delta=\{a,b,c\}$ such that $a \prec_s  b \prec_s c$. We say $\Delta$ is \emph{charged} to the edge $\{a,c\}$. For an edge $e=\{a,c\}$ such that $a \prec_s c$, $t_{e}$ denotes the number of triangles charged to edge $e$. Since every triangle is assigned to an unique edge, we can say that $T=\sum_{e \in E} t_e$.\remove{~\footnote{\comments{We emphasize that the notation $t_e$ is overloaded. While discussing multi-pass algorithm in the \al model, $t_e$ denotes the number of triangles incident on $e$. But, while discussing the one-pass  algorithm in the \al model, $t_e$ denotes the number of triangles charged to $e$}.}}

\paragraph{Intuition of sampling a light triangle:}  For the time being, let us define an edge $e$ to be \emph{light} if $t_e \leq \tau=\Theta(\sqrt{T})$~\footnote{As previously noted, the exact definition is presented in \Cref{sec:al1passheavylight}.}. Consider a light triangle $\Delta=\{u,v,w\}$ which is charged to the edge $\{u,w\}$. Let us sample an edge $e$ uniformly at random over the stream and take a triangle charged to $e$ uniformly at random with probability $\frac{t_e}{\tau}$ if $t_e \leq \tau$. The probability that $\Delta$ is sampled is $\frac{1}{m}\cdot \frac{1}{t_e}\cdot \frac{t_e}{\tau}=\frac{1}{m\tau}$. So, the probability that some light triangle will be sampled by the above procedure is $\frac{T_L}{m\tau}$, where $T_L$ denotes the total number of light triangles. If $T_L=\Omega(T)$, (otherwise, we argue that sampling light triangles is not necessary) then the success probability of sampling a light triangle is $\Omega(\sqrt{T}/m)$ and we obtain a light triangle uniformly at random. If we run the above process $\tOh(m/\sqrt{T})$ times independently, we will get a light triangle uniformly at random with high probability.

\paragraph{Intuition of sampling a heavy triangle:}  Let us define an edge $e$ to be \emph{heavy} if $t_e > \tau=\Theta(\sqrt{T})$. This implies that there are at most $\Oh(\sqrt{T})$ number of heavy edges. Let $H$ denote the set of heavy edges, and $T_H=\sum_{e\in H}t_e$. Consider the following algorithm for sampling a heavy triangle. Take an edge $e \in H$ with probability $\frac{t_e}{T_H}$ and report a triangle charged to $e$ uniformly at random. For any heavy triangle $\Delta$ (charged to $e$), the probability that $\Delta$ is sampled is $\frac{t_e}{T_H}\cdot \frac{1}{t_e}=\frac{1}{T_H}$.

Here, we would like to emphasize that both of the above intuitions can be formalized and implemented in the \al streaming model by exploiting its properties. We defer the details to \Cref{sec:al1passsubroutine}.

\section{Multi-pass Triangle Sampling in \al model
%(Proof of Theorem~\ref{theo:al3pass})
}\label{sec:almultipassalgo}

\begin{theo}[Upper Bound of Theorem~\ref{theo:two}(i)]\label{theo:tsamp-al-p-proof}
\tsamp can be solved by using $\tOh\left(m/T^{2/3}\right)$ space in \al model in $3$-passes.
\end{theo}

%For any edge $e$, let $\lambda_e$ denote the total number of triangles incident on $e$. 

Before proceeding to the algorithm, let us first define the notions of heavy and light edges as well as heavy and light triangles which will be required in the proof.

\begin{defi}[Heavy and light edge]\label{defi:heavylightedge}
Given a parameter $\tau \in \N$, an edge $e$ is said to be \emph{$\tau$-heavy} or simply heavy if $\lambda_e \geq \tau$, where $\lambda_e$ denotes the number of triangles incident on $e$. Otherwise, $e$ is said to be a \emph{$\tau$-light} or simply light edge.
\end{defi}

We now also define notions of heavy and light triangles.

%Similar to heavy and light edges,

\begin{defi}[Heavy and light triangle]\label{defi:heavylighttriangle}
Let $i, \tau$ be two integers such that $i \in \{0,1,2,3\}$. A triangle $\Delta$ is said to be \emph{$i$-light} if it has $i$ $\tau$-light edges. Note that $0$-light triangles are \emph{heavy triangles}. A triangle is said to be light if it is not heavy, that is, $i$-light for some $i \in \{1,2,3\}$.
\end{defi}

%Let us assume $\tau = {\Theta(\frac{T}{\eps^2})}^{\frac{1}{3}}$. An edge $e$ is said to be $\tau$-heavy if $\lambda_e \geq \tau$. A triangle $\Delta$ is said to be $i$-light if it has $i$ many light edges. Note that $0$-light triangles are heavy triangles.

We will use the following result which upper bounds the maximum number of triangles of a graph $G$ with respect to its number of edges.

\begin{lem}[\cite{RIVIN2002647}]\label{lem:boundnumbertriangle}
    Given a graph $G(V,E)$ such that $\size{E}=m$, the number of triangles of $G$ is at most $\Oh(m^{3/2})$.
\end{lem}

As earlier, let us denote the stream by $\cM$. Below we set $\tau = 12(T/\eps^2)^{1/3}$, where $T$ is  the number of triangles in $G$. The following observation bounds the total number of heavy edges along with the total number of heavy and light triangles.

\begin{obs}\label{obs:al3passboundno}
Let us assume that in a graph $G$, the total number of triangles is $T$ and $\tau= 12(T/\eps^2)^{1/3} $. Then the following hold:
\begin{enumerate}    
%\end{descriptio}
       \item[(i)] The number of $\tau$-heavy edges is at most $ (\eps T/8)^{2/3}$.
    
    \item[(ii)] The number of heavy triangles $T_H$ is at most $\eps T/8$.
    
    \item[(iii)] The number of light triangles $T_L$ is at least $(1-\eps /8 )T$.
\end{enumerate}
\end{obs}

\begin{proof}

\begin{enumerate}

\item[(i)] From the definition of $\tau$-heavy edges (Definition~\ref{defi:heavylightedge}), an edge $e$ is $\tau$-heavy, if the number of triangles $\lambda_e$ incident on $e$ is at least $\tau$. Since the total number of triangles in the graph $G$ is $T$, the total number of $\tau$-heavy edges can be at most $3T/\tau$. As $\tau = 12(T/\eps^2)^{1/3}$, the result follows.

\item[(ii)] 
\remove{Following the result regarding bounding the total number of triangles in a graph with respect to its number of edges (Lemma~\ref{lem:boundnumbertriangle}),} From Lemma~\ref{lem:boundnumbertriangle}, we know that the total number of possible triangles over $m$ edges is at most $\Oh(m^{3/2})$. Since from $(i)$, we know that the total number of $\tau$-heavy edges is at most $\mathcal{O}((\eps T /8 )^{2/3})$, we have the result.

\item[(iii)] Follows directly from $(ii)$ as $T=T_L+T_H$.
\end{enumerate}
\end{proof}

Now let us proceed to describe our $3$-pass algorithm in the \al model. Recall the intuition described in \Cref{sec:al-multi-over}. We will sample a multi-set of edges $\cF$ uniformly and independently at random with replacement in the first pass. The edges of $\cF$ can be thought of as \emph{candidate light edges}. Next, we count the number of triangles $\lambda_e$ incident on every edge $e \in \cF$ to decide which edges of $\cF$ are light. Additionally, for every edge $e \in \cF$, we sample a triangle $s_e=(e,e_1,e_2)$ incident on $e$ uniformly. In the third and final pass, for every edge $e \in \cF$ and every sampled triangle $s_e= (e,e_1,e_2)$, we count the number of triangles incident on $e_1$ and $e_2$. In the processing phase, every heavy edge of $\cF$ is ignored. For every light edge $e \in \cF$, we consider the triangle $s_e$ further with a suitable probability, depending upon the number of triangles incident on the edges of $s_e$. The triangle $s_e$ corresponding to some light edge $e \in \cF$ is the desired output. The formal algorithm is described in \Cref{alg:3passal}. In order to prove the correctness of \Cref{alg:3passal}, i.e. proving \Cref{theo:tsamp-al-p-proof}, we first show the following lemma about \Cref{alg:3passal}.

\begin{algorithm}[H]
\caption{Triangle Sampling in $3$-pass \al-Model}\label{alg:3passal}

%\begin{algorithmic}

\SetAlgoLined

{\bf Pass 1}: Sample a multi-set $\cF$ of $\tOh(m/T^{2/3})$ edges from $\cM$ uniformly and independently at random with replacement.\ \label{algo:3passal_line1}

{\bf Pass 2}: \For{every edge $e \in \cF$}{ \label{algo:3passal_line2}

Count $\lambda_e$.\ \label{algo:3passal_line3}

Sample a triangle $s_e$ incident on $e$ uniformly at random.\ \label{algo:3passal_line4}

}

% For every edge $e \in \cF$, perform the following\;

% \begin{enumerate}
%     \item[(i)] Count $\lambda_e$.
%     \item[(ii)] Sample one triangle $s_e$ incident on $e$ uniformly at random.
% \end{enumerate}

%\

{\bf Pass 3}: \For{every edge $e \in \cF$ and sampled triangle $s_e$ incident on $e$}{ \label{algo:3passal_line5}

%\label{algo:3passal_line6}

Let $e_1$ and $e_2$ be the other two edges of $s_e$. Compute $\lambda_{e_1}$ and $\lambda_{e_2}$.\ \label{algo:3passal_line7}

// This information along with $\lambda_e$ helps us to determine the value of $i$ such that $s_e$ is a $i$-light triangle.

}

%Let $e_1$ and $e_2$ be the other two edges of $s_e$.

% For every edge $e \in \cF$ and the sampled triangle $s_e$ incident on $e$, let $e_1$ and $e_2$ be the other two edges. Compute $t_{e_1}$ and $t_{e_2}$\;

% \

{\bf Process}: Set $\tau \leftarrow  12(T/\eps^2)^{1/3}$.\ \label{algo:3passal_line8}

\For{every edge $e \in \cF$}{ \label{algo:3passal_line9}
 
 \uIf{$e$ is $\tau$-{heavy}}{ \label{algo:3passal_line10}
            
             {\sc Ignore} $e$.\ \label{algo:3passal_line11}
            }

\uElse{ 
             Consider the triangle $s_e$ sampled on $e$.\ \label{algo:3passal_line12}

Determine $i \in \{1,2,3\}$ such that $s_e$ is an $i$-light triangle.

%            \uIf{$s_e$ is a $i$-light triangle}
%            {\label{algo:3passal_line13}

               Mark $s_e$ with probability $ \lambda_e/i \tau$.\ \label{algo:3passal_line14}
              
             % // Note that $i \neq 0$ as the triangle $s_e$ has $e$ as one of its edge and $e$ is not $\tau$-heavy.

            }
             
             % {\sc Report} $s_e$ with probability $ \frac{\lambda_e}{i \tau}$, where $s_e$ is a $i$-light triangle (Definition~\ref{defi:heavylighttriangle}).\ 

%}            

}

Report any marked triangle if there exists at least one such marked triangle. Otherwise, 
report {\sc FAIL} and {\sc ABORT} the process.\ \label{algo:3passal_line15}

%\end{algorithmic}

\end{algorithm}

\begin{lem}\label{lem:lighttriangleprob}
With high probability, Algorithm~\ref{alg:3passal} samples a light triangle uniformly. 
\end{lem}

\begin{proof}
From the description of the algorithm, note that the algorithm never reports a heavy triangle.

Now let us consider an arbitrary $i$-light triangle $\Delta$, where $i \in \{1,2,3\}$. Recall that in Pass 1, we have sampled a multi-set of edges $\cF$.  First consider the case when $\size{\cF}=1$. Note that to sample the triangle $\Delta$, we need to store one of the $i$-light edges of $\Delta$ in Pass 1. This happens with probability $i/m$. Moreover, for the sampled light edge $e$, $\Delta$ needs to be sampled among all the triangles incident on $e$, which happens with probability $1/\lambda_e$. Finally, Algorithm~\ref{alg:3passal} outputs the sampled triangle with probability $\lambda_e/i \tau$. Combining these arguments, the probability that $\Delta$ is sampled is: 
$= i/m \times 1/\lambda_e \times \lambda_e/i \tau= 1/m \tau$.

%where $\lambda_e$ denotes the number of triangles of $G$ with $e$ being an edge

Applying the union bound over all light triangles over $G$, the probability that the algorithm does not report {\sc Fail} is $T_L/m\tau$, where $T_L$ denotes the number of light triangles. Observe that, under the conditional space that the algorithm does not report {\sc Fail}, it reports a light triangle uniformly at random with probability $T_L/m\tau$ when $\size{\cF}=1$. Under the conditional space that the algorithm does not report {\sc Fail}, the probability that it outputs $\Delta$ is $1/T_L$. That is, the algorithm reports a light triangle uniformly at random with probability $T_L/m\tau$ when $|\cF|=1$.

Since $\size{\cF}= \tOh(m/T^{2/3})$ and $T_L \geq (1-\eps/8)T$, with high probability the algorithm does not report {\sc Fail} and reports a light triangle uniformly.
\end{proof}

% \begin{eqnarray*}
% &=& \frac{i}{m} \times \frac{1}{\lambda_e} \times \frac{\lambda_e}{i \tau}
% = \frac{1}{m} \cdot \frac{1}{\tau}
% \end{eqnarray*}

It is important to note that our algorithm does not sample any heavy triangle. Now we will show that the total error due to not sampling any heavy triangle is not large as $T_H \leq \eps T/8$ (from \Cref{obs:al3passboundno}).

\remove{
\begin{lem}\label{lem:heavytriangleerror}
    Total error due to not sampling any $\tau$-Heavy triangle is at most $\eps^{2/3}/2$.
\end{lem}

\begin{proof}
From the above discussion, we know that the probability of sampling a particular light triangle is $1/m \tau$. As the total number of light triangles in $G$ is $T_L$, using the union bound, we can say that the probability that a light triangle has been sampled is $T_L/m \tau$. Thus the probability of not sampling any triangle is $1- T_L/m \tau$.
    
From Observation~\ref{obs:al3passboundno} (iii), we know that $T_L \geq (1-\Oh(\eps))T$. This implies that the probability of not sampling any triangle is upper bounded by $\frac{\eps^{2/3} T^{2/3}}{m}$. As the number of edges in $\cF$ is $\Oh(m/T^{2/3})$, the total error due to heavy triangles is at most $\eps^{2/3}/2$.
\end{proof}

Now we are ready to prove the final theorem.}

\begin{proof}[Proof of Theorem~\ref{theo:tsamp-al-p-proof}]
From the description of Algorithm~\ref{alg:3passal}, the space complexity of the algorithm is $\size{\cF}=\tOh(m/T^{2/3})$. From Lemma~\ref{lem:lighttriangleprob}, Algorithm~\ref{alg:3passal} does not report {\sc Fail} with high probability, never outputs a heavy triangle and reports a light triangle uniformly. Under the conditional space that the output satisfies the guarantee by Lemma~\ref{lem:lighttriangleprob}, the distribution $\cD$ from which  Algorithm~\ref{alg:3passal} samples a triangle $\Delta$ is as follows. Let us denote the set of heavy and light triangles of $G$ as $\cT_H$ and $\cT_L$. Define
\[ \cD(\Delta)=\begin{cases} 
     1/T_L & \Delta \in \cT_L \\
      0 & \Delta \in \cT_{H}
     %  [0, (1+\eps/5)1/T_H] & \Delta \in \cT_{H,\tau'} \setminus \cT_{H,\tau}
   \end{cases}
\]

Let $\cU$ denote the uniform distribution over the set of all triangles in $G$. 
Now we bound the $\ell_1$-distance between $\mathcal{D}$ and $\mathcal{U}$ as follows:
\begin{align*}
    ||\mathcal{D} - \mathcal{U}||_1 &= \sum_{\Delta \in \cT_H} |\mathcal{D}(\Delta) - \mathcal{U}(\Delta)| + \sum_{\Delta \in \cT_L} \size{\mathcal{D}(\Delta) - \mathcal{U}(\Delta)} \leq T_H\cdot {1}/{T} + T_L \cdot \size{{1}/{T_L} - {1}/{T}}.
\end{align*}

From Observation~\ref{obs:al3passboundno}, $T_H \leq \eps T /8$ and $T_L\geq (1-\eps / 8)T$. Putting these bounds in the above expression, we have $||\cD-\cU||_1 \leq \eps$.
\end{proof}

\section{One-Pass Triangle Sampling in \al model 
%(Proof of Theorem~\ref{theo:al1pass})
}\label{sec:al1passalgo}
\noindent In this section, we give our one pass algorithm for triangle sampling in the \al model, hence proving the following theorem.

\begin{theo}[Upper bound of Theorem~\ref{theo:two}(ii)]\label{theo:al1pass_proof}
\tsamp can be solved by using $\tOh\left(m/\sqrt{T}\right)$ space in \al model in $1$-pass.
\end{theo}

To prove the above theorem, we define a notion of charging triangles to edges such that each triangle is charged to only one edge. Based on the number of triangles charged to an edge, we will define that edge to be either \emph{heavy} or \emph{light}. Informally, a triangle is heavy (light) if the edge to which it is charged is heavy (light)~\footnote{The exact definitions of heavy/light triangles and edges are based on the stream and the randomness we are using in the algorithm. The definitions are presented in \Cref{sec:al1passheavylight}.}. Consider the vertices of the graph in the order in which they are revealed in the stream. Let $a \prec_s b$ denote that the adjacency list of vertex $a$ is revealed before the adjacency list of $b$. So, the set of vertices of the graph form a total order w.r.t. $\prec_s$ relation. Consider a triangle $\Delta=\{a,b,c\}$ such that $a \prec_s  b \prec_s c$. We say that $\Delta$ is \emph{charged} to the edge $\{a,c\}$. For an edge $e=\{a,c\}$ such that $a \prec_s c$, $t_{e}$ denotes the number of triangles charged to $e$. Since every triangle is assigned to a unique edge, we  have $T=\sum_{e \in E} t_e$.

The main algorithm {\sc Triangle-Sample-\al-1pass} (Algorithm~\ref{alg:1passal}) of \Cref{theo:al1pass_proof} will be discussed in \Cref{sec:mainalgoal1pass} and \ref{sec:improve}. We discuss the  algorithm {\sc Triangle-Sample-\al-1pass} with space complexity $\tOh(m/\sqrt{T}+\sqrt{T})$. However, we can make some  modifications to the algorithm to obtain the desired space complexity  $\tOh(m/\sqrt{T})$. The details are to be discussed in \Cref{sec:improve}.
To present the main algorithm, we need two streaming subroutines to be run in parallel when the edges are appearing in the stream: \rouone (Algorithm~\ref{algo:sub1}) and \routwo (Algorithm~\ref{algo:sub2}) having space complexities $\tOh(m/\sqrt{T})$ and $\tOh(m/\sqrt{T}+\sqrt{T})$, respectively. These subroutines are described in \Cref{sec:al1passsubroutine}. 
In the post processing phase after the stream ends, {\sc Triangle-Sample-\al-1pass} has again two subroutines {\sc Sample-Light-Triangle} (Algorithm~\ref{alg:1passallight}) and {\sc Sample-Heavy-Triangle} (Algorithm~\ref{alg:1passalheavy}) which are discussed in \Cref{sec:al1passheavylight}. Note that {\sc Sample-Light-Triangle}  uses outputs from \rouone and \routwo,  where {\sc Sample-Heavy-Triangle} uses only the output of \routwo~\footnote{One may wonder that why we are using \routwo in {\sc Sample-Light-Triangle}. But we defer that discussion to \Cref{sec:al1passheavylight}.}. Finally, {\sc Triangle-Sample-\al-1pass}  combines the outputs of {\sc Sample-Light-Triangle}  and {\sc Sample-Heavy-Triangle} to report the final desired output.

%From the above description, it may seem that the space complexity of {\sc Triangle-Sample-\al-1pass}  is
%$\tOh(m/\sqrt{T}+\sqrt{T})$. 

\subsection{Descriptions of \rouone and \routwo}\label{sec:al1passsubroutine}

\noindent Here we describe the subroutines \rouone and \routwo mentioned before.

\subsubsection*{Description of \rouone}\label{sec:samplighthelper}
\noindent Recall the intuition discussed to sample a light triangle in the overview in \Cref{sec:al-onepass-over}.
To make the  intuition work, we use \rouone. It samples $\tOh(m/\sqrt{T})$ edges uniformly at random with replacement. Let $\cF_1$ be the set of sampled edges. But to count the number of edges charged to an edge $e=\{u,v\}\in \cF_1$ (such that $u \prec_s v$) and sample a triangle charged to $\{u,v\}$, we desire that $e=\{u,v\}$ is sampled when $u$ is exposed. So, by exploiting the fact that we are in the \al model, one can see all the triangles charged to $\{u,v\}$ in between the exposures of $u$ and $v$. Hence,  we can maintain the value of $t_e$ (the number of triangles charged to $e$) and a triangle charged to $e$ uniformly at random. The pseudocode for \rouone is presented in \Cref{algo:sub1} and its guarantee is stated in \Cref{lem:subroutine1}. The formal discussion about how \rouone will be helpful in sampling a light triangle will be discussed in \Cref{sec:sampheavylight}.

\begin{lem}\label{lem:subroutine1}
\rouone (\Cref{algo:sub1}) uses $\tOh(m/\sqrt{T})$ space. It reports a multiset $\cF_1$ of $\tOh(m/\sqrt{T})$ edges chosen uniformly at random with replacement from the stream. For each edge $e\in \cF_1$, the algorithm stores the following: a Boolean variable $\mathsf{flag}(e)$, a counter $t_e'$, and a triangle $\Delta_e$ containing edge $e$ if $t_e'\neq 0$~\footnote{It might be the case that $\Delta_e$ is not charged to $e$.}. Moreover, for an edge $e \in \cF_1$, if $\mathsf{flag}(e)=1$, then  (i) $t_e'=t_e$, i.e., the number of triangles  charged to $e$, and (ii) $\Delta_e$ is a triangle chosen uniformly at random among all triangles charged to $e$.
\end{lem}

\begin{proof}
As $\size{\cF_1}=\tOh(m/\sqrt{T})$ and the algorithm stores $O(1)$ information for each edge $e \in \cF_1$ (i.e., $\mathsf{flag}(e), t_e'$, and possibly $\Delta_e$), the space complexity of the algorithm follows.

Now consider the case when $\mathsf{flag}(e)=1$. Let $e=\{v,w\}$ and we have set $\mathsf{flag}(e)=1$ when $w$ is exposed in the stream. From the description of the algorithm, this implies that $e$ is already present in $\cF_1$ before $w$ is exposed, i.e, we have included $e$ into $\cF_1$ when $v$ was exposed and we have set $\mathsf{flag}(e)=0$ at that point of time and it remains $0$ until the exposure of $w$. Consider a triangle $\{v,u,w\}$ such that $v \prec_s u \prec_s w$, i.e., the triangle is charged to $e$. When $u$ is exposed, from the description of the algorithm (\Cref{algo:sub1}), we encounter the triangle $\{v,u,w\}$ and increment the counter $t_e'$.  Now consider any triangle containing $e$ but not charged to $e$, say $\{v,w,u'\}$, either $u'$ is exposed before $v$ or after $w$ in the stream. In the first case, $\mathsf{flag}(e)$ has not been initialized. In the second case, $\mathsf{flag}(e)=1$. In  either of the cases, observe that the algorithm does not consider changing $t_e'$. So, when $\mathsf{flag}(e)=1$ for $e \in \cF_1$, then $t_e'$ is in fact $t_e$. The fact that the triangle $\Delta_e$ is chosen uniformly at random among all the triangles charged to $e$ can be argued in a similar fashion.
\end{proof}

\begin{algorithm}[H]

\caption{\rouone}\label{algo:sub1}

Set $\cF_1 \leftarrow \emptyset$. Initiate a (unweighted) reservoir sampler $\mbox{{\sc RS}}$  to store (up to) $\tOh(m/\sqrt{T})$ edges.\
\label{algo:1passalsub1_line1}

Note that $\cF_1$ is the set of edges stored by {\sc RS}.

\For{every vertex $u$ when its adjacency list is revealed}{
Perform the following steps in parallel:

{\bf Step 1:} \label{algo:1passalsub1_line2}
\For{every edge $e=\{u,v\}$}{

\uIf{$\{v,u\} \in \cF_1$}{
 Set $\mathsf{flag}(e)=1$. \label{algo:1passalsub1_line3}
 
}

Give $e=\{u,v\}$ as an input to ${\sc RS}$.\ \label{algo:1passalsub1_line4}
}
%\label{algo:1passalsub1_line5}
{\bf Step 2:} \For{every edge $e=\{v,w\} \in \cF_1$}{ \label{algo:1passalsub1_line5}

\uIf{$\mathsf{flag}(e)=0$ and $u$ forms a triangle $\Delta_e$ with $e$}{

$t'_{e} \leftarrow t'_{e} + 1$.\ \label{algo:1passalsub1_line6}

Give $\Delta_e$ as input to {\sc RS} that samples a triangle charged to $e$ uniformly at random. \label{algo:1passalsub1_line7}

} 

}

\color{red}

% Let $\cF_1$ and $\cF'_1$ be the set of sampled edges by $\mbox{{\sc RS}}$ before and after the adjacency list of $u$ has been revealed, respectively.

% \For{every edge $e \in \cF_1 \setminus \cF'_1$}{

% Delete all information corresponding to $e$.\ \label{algo:1passalsub1_line8}
% }

% \For{every edge $e \in \cF'_1 \setminus \cF_1$}{

% Set $\mathsf{flag}(e)=0$ and initiate an (unweighted) reservoir sampler to sample a triangle charged on $e$ uniformly at random. \label{algo:1passalsub1_line9}

% }

\color{black}

% \remove{
% {\bf Process (at the end of the stream):}

% Set $\tau \leftarrow \left(\frac{1+\eps/10}{1-\eps/10}\right)\sqrt{T/\eps}$.

% \uIf{$\mathsf{flag}(e_s)=1$ and $0<t_{e_s} \leq \tau$}
% {
% Report $\Delta_s$ with probability $\frac{t_{e_s}}{\tau}$. With remaining probability, Return {\sc Fail}.\
% }

%   \uElse{
      
%       Return {\sc Fail}.
%    }   }
}

\For{each edge $e \in \cF_1$}
{
Report $\mathsf{flag}(e),t_e'$, and $\Delta_e$.

}

\end{algorithm}

\subsubsection*{Description of \routwo}\label{sec:sampheavyhelper}

\noindent Recall the intuition discussed to sample a heavy triangle in the overview in \Cref{sec:al-onepass-over}.
To make the intuition work, we use \routwo. Note that it stores a set $\cF_2$ of edges where each edge is included in $\cF_2$ independently with probability $p=100\log n/\eps^2\sqrt{T}$.
First, let us discuss how we essentially detect all the heavy edges. Consider a heavy edge $e=\{u,v\}$ such that $u \prec_s v$. If $\{u,v\}$ is included in $\cF_2$ when $u$ is exposed, then one can compute $t_e$ exactly (as we do in \rouone).
But it may be the case that $e$ is  not included in $\cF_2$ when $u$ is exposed as we are including it in $\cF_2$ with probability $p$. However, we can guarantee that a \emph{good fraction} of edges of the form $\{u,w\}$ are included in $\cF_2$ (with high probability) such that $\{u,v,w\}$ is a triangle charged to $e$. When $\{u,v\}$ arrives in the stream for the second time (during the exposure of $v$), we can detect the set $S_e$ of such triangles $\{u,v,w\}$ such that $e=\{u,w\}\in \cF_2$. Due to the properties of the \al model, we can find the set $S_e$ for every heavy edge $e$. We can argue that $|S_e|/p$ approximates $t_e$ with high probability. Moreover,  we show that a triangle taken uniformly at random from $S_e$ is essentially similar to a triangle chosen uniformly at random from all the triangles charged to $e$. The pseudocode for \routwo is presented in \Cref{algo:sub2} and its guarantee is stated in \Cref{lem:subroutine2}. The formal discussion about how \routwo will be helpful in sampling a heavy triangle will be discussed in \Cref{sec:sampheavylight}. 

\begin{algorithm}[H]
\caption{\routwo}\label{algo:sub2}
Set $\cF_2 \leftarrow \emptyset$, and $\kappa=10\log n$, $H \gets \emptyset$.\ \label{algo:1passalsub2_line1}

%Set $\kappa=\left(1-\frac{\eps}{10}\right) \frac{ \tau \log n}{\sqrt{T}}$.\

%\textcolor{red}{Initialize a weighted reservoir sampler {\sc WRS}.}\ \label{algo:1passalsub2_line2}

\For{every vertex $u$ when its adjacency list is revealed}
{

Perform Steps 1 and 2 in parallel:

{\bf Step 1:} \For{every edge $e=\{u,v\}$ in the stream}
{\label{algo:1passalsub2_line3}
Set $x_e=0$. \label{algo:1passalsub2_line4}

 $H=H \setminus \{e\}$.

If {$e \in \cF_2$}, 
Set $\mathsf{flag}(e)=1$.\ \label{algo:1passalsub2_line5}

\uElse{

Set $\mathsf{flag}(e)=0$.\ \label{algo:1passalsub2_line6}

}

Include $\{u,v\}$ in $\cF_2$ with probability $p= 100 \log n/\eps^2 \sqrt{T}$.

}
%\label{algo:1passalsub2_line9}
{\bf Step 2:} \For{each edge $e = \{x,w\} \in \cF_2$ such that $\mathsf{flag}(e)=1$}{ 

\label{algo:1passalsub2_line9}

Determine if $\{u,x,w\}$ is a triangle.\ \label{algo:1passalsub2_line10}

$x_{e}\leftarrow x_{e}+1$. \label{algo:1passalsub2_line11}

}

{\bf Processing information obtained from Step 1 and 2:}

After the adjacency list of $u$ gets exposed, for each edge $e$ revealed during the exposure of $u$, do the following: 

\uIf{${x}_e \geq \kappa$}{

Include $e$ in $H$.

Choose a triangle $\Delta_e$  uniformly at random among the set of triangles charged to $e$ and detected in {\bf Step 2}.

}

\color{red}

% \uIf{$x_e <\kappa$}{ 

% Set $light(e)=1$.\ \label{algo:1passalsub2_line12}

% }

% $\widehat{t}_e=\frac{1}{p} \cdot x_e$.\ \label{algo:1passalsub2_line13}

% $\widehat{T}_H \leftarrow  \widehat{T}_H + \widehat{t}_e$.\ \label{algo:1passalsub2_line14}
\color{black}
}
\For{each edge $e \in H$}
{
report $x_e$ and $\Delta_e$.

}

\end{algorithm}
%The formal discussion how \routwo will be helpful in sampling a heavy triangle will be discussed in \Cref{sec:sampheavylight}.

\begin{lem}\label{lem:subroutine2}
\routwo (\Cref{algo:sub2}) uses $\tOh\left({m}/{\sqrt{T}}+\sqrt{T}\right)$ space in expectation. It reports a set of edges $\cF_2$ where each $e \in E$ is included in $\cF_2$ with probability $p={100\log n}/{\eps^2 \sqrt{T}}$ independently.  The algorithm stores a set of heavy edges $H$. Also, for each  $e \in H$, it stores a random variable $x_e$ (to possibly estimate $t_e$, i.e, the number of triangles charged to $e$), and a triangle $\Delta_e$ charged to edge $e$ if $x_e \neq 0$. Moreover,  $x_e/p$ is a $(1\pm \eps/30)$-approximation to $t_e$. Also, for any triangle $\Delta$ (charged to $e$), we have $\Pr(\Delta_e=\Delta)$ = $(1\pm \eps/15)\cdot 1/t_e$.
\end{lem}

\begin{proof}[Proof of \Cref{lem:subroutine2}]
As each edge is included in $\cF_2$ with probability $p= 100 \log n/\eps^2 \sqrt{T}$, $\size{\cF_2}=100 m \log n/\eps^2 \sqrt{T}$ in expectation. For each edge $e \in H$, the algorithm stores $x_e$ and possibly $\Delta_e$. Later in the proof, we argue that $|H|=\Oh({\sqrt{T}})$ with high probability. So, the expected space complexity is bounded by $\tOh(\size{\cF_2}+\size{H})= \tOh\left({m}/{\sqrt{T}}+\sqrt{T}\right)$.

Now, we argue that for an edge $e=\{u,v\} \in H$, $x_e/p$ is a $(1 \pm \eps/30)$-approximation to $t_e$ with high probability. Assume the adjacency list of $v$ is revealed before the adjacency list of $u$. Let $X$ be the set of triangles charged to $e=\{u,v\}$ and ${X}'\subset X$ detected in {\bf Step 2} of \routwo when the adjacency list of $u$ is revealed. For any triangle $\Delta=\{u,v,w\}$ charged to $e=\{u,v\}$, $\Delta \in X'$ if the edge $\{v,w\}$ has been added to $\cF_2$ when the adjacency list of $v$ was revealed. So, the probability that the triangle $\{u,v,w\}$ is in $X'$ is $p$. Note that $x_e=\size{X'}$ and $t_e=\size{X}$.% $x_e \geq \kappa$. Hence, $\size{X} \geq \tau/10$. By Chernoff bound, $\frac{1}{p}\size{X'}$ is an $(1\pm \eps/30)$-approximation to $\size{X}=t_e$.

Using Chernoff bound, we can show that for any edge $e$, if $t_e \geq \tau=900\sqrt{T}$, then with high probability, $x_e \geq\kappa$, and, if $t_e < \tau/10$, then $x_e < \kappa$ with high probability.

Consider an edge $e$ with $x_e \geq \kappa$. Note that $t_e \geq \tau/10$ with high probability. By Chernoff bound, $\widehat{t}_e={x}_e/p$ is a $\left(1 \pm \eps/30\right)$-approximation to $t_e$ with high probability. 

From the description of the algorithm,  we take a triangle $\Delta_e$ uniformly at random from the set $X'$. Recall that ${X}'$ is the set of triangles charged to $e$ that are detected in {\bf Step 2} of \routwo when the adjacency list of $u$ is revealed.  For any triangle  $\Delta$ charged to $e=\{u,v\}$, we can deduce the following:
$$\pr(\Delta_e=\Delta)=\Pr(\Delta \in X') \cdot \frac{1}{\size{X'}}.$$

Let the other vertex in the triangle $\Delta$ be $x$. So, $\Delta$ is included in $X'$ during the exposure of $u$ if the edge $\{v,x\}$ is included in $\cF_2$ during the exposure of $v$. So, $\pr(\Delta \in X')=p$. Also, note that $\size{X'}/p = x_e/p=\widehat{t}_e$ is a $(1\pm \eps/30)$-approximation to $t_e$ with high probability.  

Now, we bound the size of $H$. For each edge $e \in H$, $x_e \geq \kappa$, i.e, $t_e \geq \tau/10$ with high probability. As each triangle is charged to a unique edge, $|H|$ is bounded by $10T/\tau \leq \Oh(\sqrt{T})$. Hence, we are done with the proof of the lemma.
\end{proof}

\subsection{Notions of heavy and light edges and triangles}\label{sec:al1passheavylight}
\noindent
Note that, in the overview in \Cref{sec:al-onepass-over}, we defined an edge $e$ to be either heavy or light based on the value of $t_e$, i.e., the exact number of triangles charged to $e$. For a heavy edge $e$, as \routwo does not determine $t_e$ exactly, the notions of heavy and light edges in terms of actual number of triangles charged to them would not be helpful in the analysis. To cope up with that, we define an edge $e$ to be heavy or light based on the value of $x_e$ in \routwo.

%Let us consider random variable $x_e$ for every edge $e$ as considered by \routwo. Based on $x_e$, we define the following:

\begin{defi}[Heavy and light edge]
Let $e=(u,v)$ such that $u \prec_s v$. $x_e$ is the random variable that denotes the number of triangles that are charged to $e$, and is detected by \routwo (Algorithm~\ref{algo:sub2}), when the adjacency list of $v$ is revealed. 
Let $\kappa \in \N$ be a parameter. An edge $e \in E$ is said to be \emph{$\kappa$-heavy} or simply \emph{heavy} if $x_e \geq \kappa$. Otherwise, $e$ is said to be \emph{$\kappa$-light} or simply \emph{light}.
\end{defi}

Now we define the notion of heavy and light triangles.

\begin{defi}[Heavy and light triangle]
Let $\kappa\in \N$ and $\Delta$ be a triangle in the graph $G$, which is charged to the edge $e$. $\Delta$ is said to be $\kappa$-heavy if $e$ is $\kappa$-heavy or simply \emph{heavy}. Otherwise, $\Delta$ is $\kappa$-light or simply \emph{light}.
    %Let $a,b,c \in V$ be three vertices of $G$ such that $a \prec_s b \prec_s c$, and $\Delta = (a,b,c)$ be a triangle of $G$. For a parameter $\tau \in \N$, the triangle $\Delta$ is said to be \emph{$\tau$-heavy} if $t_\{a,c\} \geq \tau$. Otherwise, $\Delta$ is said to be a \emph{$\tau$-light} triangle.
\end{defi}

Let $\cT_L$ and $\cT_H$ denote the sets of $\kappa$-light and $\kappa$-heavy triangles in $G$, respectively, and $T_L$ and $T_H$ be the number of $\kappa$-light and $\kappa$-heavy triangles in $G$, respectively. Thus $T_L=\size{\cT_L}$ and $T_H=\size{\cT_H}$ and $T=T_L +T_H$. 

%Let us now describe the algorithm.

\subsection{Descriptions of \sampllight and \samplheavy}\label{sec:sampheavylight}
%\color{blue}
In this section, we describe the algorithms \sampllight and \samplheavy that essentially sample a $\kappa$-light edge and $\kappa$-heavy edge, respectively.
The two algorithms (\sampllight and \samplheavy) are not streaming algorithms in nature, they are working on the outputs produced by \rouone and \routwo (which are both working on the stream in parallel).

\color{black}
\subsubsection*{Description of \sampllight}\label{sec:samplight}

\begin{algorithm}[H]
\caption{{\sc Sample-Light-Triangle}}\label{alg:1passallight} 

%\begin{algorithmic}

\SetAlgoLined
%Execute \rouone and \routwo over the stream. 

%\textbf{Process at the end of the stream:}

Let $\cF_1$ be the set of edges we get from \rouone satisfying \Cref{lem:subroutine1}. Also, let $H$ be the set of heavy edges we get from \routwo satisfying \Cref{lem:subroutine2}.

%For each edge $e \in \cF_1$,

%Set $\tau \leftarrow \left(\frac{1+\eps/10}{1-\eps/10}\right)\sqrt{T/\eps}$.
Set $\tau=900 \sqrt{T}, \mbox{{\sc Sampled-Triangle}} \leftarrow \emptyset$.\ \label{alg:1passallight_line1}

\For{every edge $e \in \cF_1$ with $\mathsf{flag}(e)=1$ and $e \notin H$}{
\label{alg:1passallight_line2}

\uIf{$t_e \leq \tau$}{

Keep the  triangle $\Delta_e$ charged to $e$ reported by \rouone. 
%corresponding reservoir sampler (if it exists).
%\label{alg:1passallight_line3}

$\mbox{\sc Sampled-Triangle} \leftarrow \mbox{\sc Sampled-Triangle} \cup \Delta_e$.

}
} 

\uIf{$\mbox{{\sc Sampled-Triangle}} \neq \emptyset$}{ 

Report one triangle in {\sc Sampled-Triangle}  arbitrarily.\ \label{alg:1passallight_line4}
}

Otherwise, return {\sc Fail}. \label{alg:1passallight_line5}

%\end{algorithmic}

\end{algorithm}

\begin{rem}\label{rem:light-heavy}
    Here we would like to discuss why {\sc Sample-Light-Triangle} uses \routwo along with \rouone.  Note that $H$ is the set of $\kappa$-heavy edges, as determined by the algorithm \routwo. Any edge not in $H$ is the set of $\kappa$-light edges. It may be the case that, for an edge $e \in \cF_1$, $t_e\leq \tau$. But due to the approximation error in \routwo, $x_e \geq \kappa$, i.e., $e \in H$. To take care of the fact, we consider only those edge $e\in \cF_1$ such that $\mathsf{flag}(e)=1$ and $e \notin H$.
\end{rem}

In the following lemma, we prove that Algorithm~\ref{alg:1passallight} samples a $\kappa$-light triangle uniformly.

\begin{lem}\label{lem:light}
Consider the algorithm {\sc Sample-Light-Triangle} (Algorithm~\ref{alg:1passallight}).  If $T_L=\Omega(\eps T)$, then with high probability, the algorithm samples a $\kappa$-light triangle uniformly, that is, with probability $T_L/2 m \tau$, where $T_L$ denotes the number of $\kappa$-light triangles of $G$.
% Moreover, the algorithm uses constant space.
\end{lem}

\begin{proof}
%Note that the space complexity of \Cref{alg:1passallight} follows from the guarantee of \rouone.

%First we argue that, when we take just one edge uniformly at random in $\cF_1$, then {\sc Sample-Light-Triangle} reports a light triangle uniformly at random with probability $T_L/2 m \tau$. 

%From the description of the algorithm, it is clear that Algorithm~\ref{alg:1passallight} uses constant space. Now let us argue the fact that Algorithm~\ref{alg:1passallight} samples $\tau$-light triangles uniformly with the claimed probability.
Recall that $\cF_1$ is the set of edges sampled in \rouone. Consider an edge $e_s$ in $\cF_1$. For the time being assume that $\size{\cF_1}=1$. Observe that $\mathsf{flag}(e_s)=1$ if $e_s$ is included in $\cF_1$ when $e_s$ arrives for the first time. By \Cref{lem:subroutine1}, if $\mathsf{flag}(e_s)=1$, then $t_{e_s}'=t_{e_s}$ and $\Delta_{e_s}$ is a uniform triangle charged to $e$. Consider a $\kappa$-light triangle $\Delta=\{a,b,c\}$ which is charged to $\{a,c\}$. 
\begin{eqnarray*}
    \pr\left(\Delta_{e_s}=\Delta \right) &=& \pr(e_s =\{a,c\}~\mbox{and}~\mathsf{flag}(e_s)=1) \times \ \pr(\Delta_{e_s}=\Delta ~|~e_s =\{a,c\}~\mbox{and}~\mathsf{flag}(e_s)=1) \\
    &=& \frac{1}{2m} \times \frac{1}{t_e} 
\end{eqnarray*}

Let $\Delta_L$ be the output of {\sc Sample-Light-Triangle}. As we include $\Delta_{e_s}$ to {\sc Sampled-Triangle} with probability  $t_e/\tau$,
we can say that:
$$\pr\left(\Delta_{L}=\Delta \right)=\pr\left(\Delta_{e_s}=\Delta \right)\cdot \frac{t_e}{\tau}=\frac{1}{2m \tau}$$

The probability that Algorithm~\ref{alg:1passallight} reports some $\kappa$-light triangle is $T_L/2m\tau$ which follows by applying the union bound over all light triangles. Observe that, under the conditional space that Algorithm~\ref{alg:1passallight} does not report {\sc Fail}, it returns a $\kappa$-light triangle uniformly. Since $\size{\cF_1}=\tOh(m/\sqrt{T})$ and $T_L=\Omega(\eps T)$, with high probability, {\sc Sample-Light-Triangle} reports a light triangle uniformly at random.
\end{proof}

%\subsection{\samplheavy}

\subsubsection*{Description of \samplheavy}\label{sec:sampheavy}

%In this section, we describe the algorithm \samplheavy, which uses the subroutine \routwo.

\begin{algorithm}[H]
\caption{{\sc Sample-Heavy-Triangle}}\label{alg:1passalheavy} 

%Set $p= 100 \log n/\eps^2 \sqrt{T}$. 

%Execute \routwo over the stream.

%\textbf{Process at the end of the stream:}

Let $H$ be the set of heavy edges we get from \routwo satisfying \Cref{lem:subroutine2}.

\uIf{$H= \emptyset$}{

Report ABORT and return FAIL.

}

Choose an edge $r \in H$ proportional to its weight, where the weight of an edge $e \in H$ is defined as $\frac{x_e}{\sum\limits_{e' \in H}x_{e'}}$.

%Choose $r \in H$ randomly such that the probability of $r=e $ is $\frac{x_e}{\sum\limits_{e'\in H}x_{e'}}$.

Let  $\Delta_r$ be the triangle charged to the edge $r$, as reported by \routwo.

Report $\Delta_r$ as the output $\Delta_H$ (as the sampled heavy triangle).

%$\widehat{T}_H \leftarrow  \sum\limits_{e \in H}\frac{x_e}{p}$.

% \For{each $e \in H$}{

% Report $x_e$ and $\Delta_e$.

% }

% \color{red}

% \uIf{{\sc WRS} outputs $({x}_e,\Delta_e)$}{ 
% \label{alg:1passalheavy_line1}

% Report $\Delta_e$.\ \label{alg:1passalheavy_line2}

% }

% \uElse{

% Report {\sc Fail}.\ \label{alg:1passalheavy_line3}

% }

\color{black}

% \uIf{{\sc WRS} does not output anything}{

% Report {\sc Fail}.\

% }

% \uElse{

% $({x}_e,\Delta_e) \leftarrow {\sc WRS}$.\

% $\Delta_H \leftarrow \Delta_e$.\

% Report $\Delta_H$.\

% }
\end{algorithm}

%Now we show that the estimate of the number of heavy triangles $\widehat{T}_H$ found by \routwo is close to $T_H$, the number of heavy triangles of $G$. Also, we show that {\sc Sample-Heavy-Triangle} (Algorithm~\ref{alg:1passalheavy}) returns a $\kappa$-heavy triangle almost uniformly.

Let us begin by showing that the estimate of the number of heavy triangles $\widehat{T}_H$ found by \routwo is close to $T_H$, the actual number of heavy triangles of $G$.

\begin{coro}\label{lem:heavytriangleapproximate}
With high probability, $\widehat{T}_H=\frac{1}{p}\sum\limits_{e \in H}x_e$ is a $\left(1 \pm \eps/30\right)$-approximation to $T_H$.    
\end{coro}

\begin{proof}
From \Cref{lem:subroutine2}, we know that for a $\kappa$-heavy edge $e$ (that is when $x_e\geq \kappa$),  $\widehat{t}_e={x}_e/p$ is a $\left(1 \pm \eps/30\right)$-approximation to $t_e$. Since $\widehat{T}_H=\sum_{e \in H}x_e/p$, this implies that $\widehat{T}_H$  is a $\left(1 \pm \eps/30\right)$-approximation to $T_H$.
\end{proof}

Now we show that {\sc Sample-Heavy-Triangle} (Algorithm~\ref{alg:1passalheavy}) returns a $\kappa$-heavy triangle almost uniformly.

\begin{lem}\label{lem:heavy--al-one}
Consider Algorithm~\ref{alg:1passalheavy}.
Let us denote the output produced by {\sc Sample-heavy-Triangle} as $\Delta_H$, unless it returns {\sc Fail}. Then, with high probability, $\Delta_H$ follows a distribution $\cD_H$ over the set of $\kappa$-heavy triangle  such that $\cD_H(\Delta)=\pr(\Delta_H=\Delta) = (1\pm \eps/5)/T_H$ for any $\kappa$-heavy triangle $\Delta$.
    
%    {\item[(iii)]  For any $\Delta \in \cT_{H,\tau} \setminus \cT_{H,\tau'}$, $\pr(\Delta_s=\Delta) \leq (1+\eps/5)/T_H$.}
    
    %\item[(iv)] For any $\kappa$-Light triangle $\Delta,$ $\pr(\Delta_s=\Delta)=0$.

    %Consider Algorithm~\ref{alg:1passalheavy} and its output set $F'$. With high probability, each edge $e \in F'$ is a $\tau$-Heavy edge and all $\left(\frac{1+\eps/10}{1-\eps/10}\right)\tau$-Heavy edges of $G$ are in $F'$. However, some $\tau$-Heavy edges may not be present in $F'$. Moreover, the space complexity of Algorithm~\ref{alg:1passalheavy} is $\tOh\left(\frac{m}{\sqrt{T}}+\sqrt{T}\right)$.
\end{lem}

\begin{proof}
Consider a $\kappa$-heavy edge $e=\{u,v\}$ such that the adjacency list of $u$ is revealed sometime after the adjacency list of $v$ is revealed. Let $\cE_e$ be the event that  we report some triangle charged to $e$. From the description of \routwo, we can say that  
$$
    \pr(\cE_e)=\frac{x_e}{\sum\limits_{e\in H}{x_e}}=\frac{x_e/p}{\sum\limits_{e\in H}{x_e/p}}=\frac{\widehat{t}_e}{\widehat{T}_H}=\left(1 \pm \frac{\eps}{10}\right)\frac{t_e}{T_H}.
$$

This is because $\widehat{t}_e$ and $\widehat{T}_H$ are $(1\pm \eps/30)$-approximations of $t_e$ and $T_H$, respectively, which follows from \Cref{lem:subroutine2} and \Cref{lem:heavytriangleapproximate}. 
Now consider a triangle $\Delta \in \cT_H$ and let us calculate that our algorithm reports $\Delta$ as the output $\Delta_H$. Let $e$ be the edge to which $\Delta$ is charged. Observe the following:
\begin{align*}\label{eqn:one}
    \pr(\Delta_H=\Delta)&=\pr(\cE_e) \cdot \pr(\Delta_H=\Delta~|~\cE_e) & \\ &= \left( 1\pm \frac{\eps}{10} \right) \frac{t_e}{T_H} \cdot \pr(\Delta_H=\Delta~|~\cE_e) & \\
    &= \left( 1\pm \frac{\eps}{10} \right) \frac{t_e}{T_H} \cdot \pr(\Delta_e=\Delta~|~\cE_e) & \\
    &= \left( 1\pm \frac{\eps}{10} \right) \frac{t_e}{T_H} \cdot \left(1\pm \frac{\eps}{15}\right)\cdot \frac{1}{t_e} &[\mbox{From  \Cref{lem:subroutine2}}]\\
    &= \left(1 \pm \frac{\eps}{5}\right) \cdot \frac{1}{T_H} &
\end{align*}
\end{proof}

\subsection{ One-pass algorithm with space complexity $\tOh(m/\sqrt{T}+\sqrt{T})$.}\label{sec:mainalgoal1pass}

\noindent In this section, we discuss {\sc Triangle-Sample-\al-1pass} (Algorithm~\ref{alg:1passal}) as our one-pass algorithm for \tsamp in the \al model, with space complexity $\tOh(m/\sqrt{T}+\sqrt{T})$. {\sc Triangle-Sample-\al-1pass} runs \rouone (Algorithm~\ref{algo:sub1}) and \routwo (Algorithm~\ref{algo:sub2}) in parallel during the stream. 
In the post processing phase after the stream ends, {\sc Triangle-Sample-\al-1pass}  has again two subroutines {\sc Sample-Light-Triangle} (Algorithm~\ref{alg:1passallight}) and {\sc Sample-Heavy-Triangle} (Algorithm~\ref{alg:1passalheavy}) which essentially samples $\kappa$-light and $\kappa$-heavy triangle, as guaranteed in \Cref{sec:samplight}. Finally, {\sc Triangle-Sample-\al-1pass}  combines the outputs of {\sc Sample-Light-Triangle}  and {\sc Sample-Heavy-Triangle} to report the final desired output.

\begin{algorithm}[H]
                \caption{Triangle Sampling in 1-pass \al Model}\label{alg:1passal} 
%Call {\sc Sample-Light-Triangle}. 

%Initiate a reservoir sampler $\mbox{{\sc RS}}$ such that $\mbox{{\sc RS}}$ keeps a set $\cF_1$ of $\tOh(m/\sqrt{T})$ edges. Also initiate a weighted reservoir  sampler {\sc WRS}.\ \label{algo:1passal_line1}

Set $p=100\log n/\eps^2 \sqrt{T}$ and  $\kappa= 10 {\log n}$.\ \label{algo:1passal_line2}

%// {We assume that $T$ is at least a polynomial in $\log n$ and $1/\eps$.} So, $p \in (0,1)$ .

Call \rouone and \routwo  in parallel.\ \label{algo:1passal_line3}

{\bf Process at the end of the stream:}

Consider the set $H$ of heavy edges and the value of $x_e$ for each edge $e \in H$, as reported by \routwo. 

Compute $\widehat{T}_H=\sum\limits_{e\in H}\widehat{t}_e$, where $\widehat{t}_e=\frac{1}{p}x_e$.

Call {\sc Sample-Light-Triangle}. Let $\Delta_L$ be the output.\ \label{algo:1passal_line4}

Call {\sc Sample-Heavy-Triangle}. Let $\Delta_H$ be the output.\ \label{algo:1passal_line5}

If either {\sc Sample-Light-Triangle} or {\sc Sample-Heavy-Triangle} report {\sc Fail}, then report {\sc Fail} and {\sc ABORT} the process.\ \label{algo:1passal_line6}

%If it reports {\sc Fail}, then report the same as the output. Otherwise, let $\widetilde{T_h}$ and $\Delta_h$ be its output.

\uIf{$\widehat{T}_H \leq \eps T/10 $}
{ 
\label{algo:1passal_line7}

Report $\Delta_L$ as the output.\ \label{algo:1passal_line8}

}

\uIf{$\widehat{T}_H \geq \left(1-\eps/10\right) T$} 
{

\label{algo:1passal_line9}

Report $\Delta_H$ as the output.\ \label{algo:1passal_line10}

}

Report $\Delta_H$ with probability $\widehat{T}_H/T$ and $\Delta_L$ with  probability $1-\widehat{T}_H/T$.\
\label{algo:1passal_line11}
\end{algorithm}

In \Cref{theo:al-one}, we show that the space complexity of \cref{alg:1passal} is $\tOh(m/\sqrt{T}+\sqrt{T})$ and show that \cref{alg:1passal} solves \tsamp in \al model with high probability.

\begin{theo}\label{theo:al-one}
Consider Algorithm~\ref{alg:1passal}. It uses $\tOh\left({m}/{\sqrt{T}}+\sqrt{T}\right)$ space in expectation. Moreover, it reports a triangle from a distribution $\cD$ over the set of triangles in $G$ such that $||\cD-\cU||_1 \leq \eps$, where $\cU$ denotes the uniform distribution over the set of all triangles in $G$.

\end{theo}

\begin{proof}
The space complexity follows from \Cref{lem:subroutine1} and \Cref{lem:subroutine2}. In order to prove the other part, we break the analysis into three parts.
\paragraph*{\textbf{Case 1: $\widehat{T}_H \leq \frac{\eps}{10}T$:}}
In this case, we report the output $\Delta_L$ produced by {\sc Sample-Light-Triangle}. From Lemma~\ref{lem:light}, for any light triangle $\Delta \in \cT_L,$ $ \pr(\Delta_L=\Delta)=\remove{(1 \pm \eps/10)}\frac{1}{T_L}.$
\noindent
Note that, in this case the algorithm never reports a triangle that is not in $\cT_L$. So, 
\[ \cD(\Delta)=\begin{cases} 
      1/T_L\remove{\left[(1-\eps/10)\frac{1}{T_L}, (1+\eps/10)\frac{1}{T_L}\right]} & \Delta \in \cT_L \\
      0 & \Delta \in \cT_H
   \end{cases}
\]
Thus,
\begin{equation}\label{eqn:aa}
||\cD-\cU||_1=T_L \cdot \size{{1}/{T_L}-{1}/{T}} \remove{\max \left \lbrace  \size{{(1-\eps/10){1}/{T_L}-{1}/{T}}}, \size{{(1+\eps/10)\frac{1}{T_L}-\frac{1}{T}}}  \right \rbrace}+T_H \cdot {1}/{T}.
\end{equation}

As $\widehat{T}_H$ is a $\left(1\pm\eps/30\right)$-approximation of $T_H$, in this case ${T_H} \leq \frac{\frac{\eps}{10}T}{1-\eps/30} \leq \eps T/8$. So, $\left(1-\eps/8\right)T \leq T_L \leq T$. Combining the bounds on $T_L$ and $T_H$ in Equation~\eqref{eqn:aa}, we conclude that $||\cD-\cU||_1 \leq \eps$.

\paragraph*{\textbf{Case 2: $\widehat{T}_H \geq \left(1-\frac{\eps}{10}\right)T$:}} 
In this case, we report the output $\Delta_H$ produced by {\sc Sample-Heavy-Triangle}. From Lemma~\ref{lem:heavy--al-one}, for any  triangle $\Delta \in \cT_{H}$, we have 
$\pr(\Delta_H=\Delta) =  (1\pm\eps/5)/T_H$.
Note that, here the algorithm never reports a triangle that is in $\cT_L$. So, 
\[ \cD(\Delta)=\begin{cases} 
     0 & \Delta \in \cT_L \\
      (1\pm \eps/5)/T_H & \Delta \in \cT_{H}
     %  [0, (1+\eps/5)1/T_H] & \Delta \in \cT_{H,\tau'} \setminus \cT_{H,\tau}
   \end{cases}
\]
Thus,
\begin{equation}\label{eqn:bb}
\small
||\cD-\cU||_{1} = \frac{T_{L}}{T} + T_H \cdot\max \left \lbrace \size{\left(1-\frac{\eps}{5}\right)\frac{1}{T_H}-\frac{1}{T}}, \size{\left(1+\frac{\eps}{5}\right)\frac{1}{T_H}-\frac{1}{T}} \right \rbrace.
\end{equation}

As $\widehat{T}_H$ is a $\left(1\pm\eps/30\right)$-approximation of $T_H$, ${T_H} \geq \frac{(1-\eps/10)}{1+\eps/30}T \geq (1-\eps/5)T$. So, $T_L \leq \eps T/5$.  Putting the bounds on $T_L$ and $T_H$ in Equation~\ref{eqn:bb}, we have $||\cD-\cU||_1 \leq \eps$.

\paragraph*{\textbf{Case 3: $\frac{\eps}{10}T<\widehat{T}_H < \left(1-\frac{\eps}{10}\right)T$:}}
In this case, we report $\Delta_H$ produced by {\sc Sample-Heavy-Triangle} with probability $\widehat{T}_H/T$ and $\Delta_L$ produced by {\sc Sample-Light-Triangle} with probability $1-\widehat{T}_H/T$. Let $\Delta_s$ be the triangle reported  by the algorithm. So, for any light triangle $\Delta \in \cT_L$, we get 
$$\pr(\Delta_s=\Delta)=\left(1-\frac{\widehat{T}_H}{T}\right) \times \frac{1}{T_L}$$
From Lemma~\ref{lem:heavy--al-one}, for any heavy triangle $\Delta \in \cT_{H}$, we have:
$$ \pr(\Delta_s=\Delta) =  \frac{\widehat{T}_H}{T} \cdot (1 \pm \frac{\eps}{5})  \frac{1}{{T_H}}.$$

%For any triangle $\Delta \in \cT_{H,\tau} \setminus \cT_{H,\tau'}$, $$ \pr(\Delta_s=\Delta) \leq (1+\eps/5)\frac{\widehat{T}_H}{T} \cdot \frac{1}{{T_H}}.$$

\[ \cD(\Delta)=\begin{cases} 
     \left(1 - \frac{\widehat{T}_H}{T}\right) \cdot \frac{1}{T_L} & \Delta \in \cT_L \\
       \frac{\widehat{T}_H}{T} \cdot (1 \pm \frac{\eps}{5})  \frac{1}{{T_H}} & \Delta \in \cT_{H}
%       \left[0, (1+\eps/5)\frac{\widehat{T}_H}{T} \cdot \frac{1}{{T_H}}\right] & \Delta \in \cT_{H,\tau'} \setminus \cT_{H,\tau}
   \end{cases}
\]

So, 
\begin{align*}\label{eqn:cc}
||\cD-\cU||_{1}
&= T_L \cdot \size{\frac{1}{T}-\left(1-\frac{\widehat{T}_H}{T}\right)\frac{1}{T_L}} \\ 
& \;\; \;\; \;\; \;\;
+ {T_{H}} \cdot \max \left \lbrace \size{\left(1-\frac{\eps}{5}\right)\frac{1}{T_H} \cdot \frac{\widehat{T}_H}{T}-\frac{1}{T}}, \size{\left(1+\frac{\eps}{5}\right)\frac{1}{T_H} \frac{\widehat{T}_H}{T}-\frac{1}{T}} \right \rbrace .
\end{align*}

As $\widehat{T}_H$ is a $\left(1\pm\eps/30\right)$-approximation of $T_H$, in this case $ \frac{\eps T/10}{1+\eps/30}\leq {T_H} \leq T$. Also, $T_L \leq T$.   Combining these bounds in the above expression, we conclude that $||\cD-\cU||_1 \leq \eps$.
\end{proof}

\subsection{Improving the space complexity of \Cref{alg:1passal} to $\tOh(m/\sqrt{T})$}\label{sec:improve}

\noindent The space complexity $\tOh(m/\sqrt{T}+\sqrt{T})$ of \Cref{alg:1passal} (as stated in \Cref{theo:al-one}) is due to the same space complexity of \routwo (as stated in \Cref{lem:subroutine2}). The space complexity $\tOh(m/\sqrt{T}+\sqrt{T})$ of \routwo is due to the fact that we store the set of heavy vertex $H$ in \routwo and we argue that $|H|=\Oh(\sqrt{T})$. Here, we argue that we do not need to store the set $H$ explicitly. By Line 17--19 of \routwo (\Cref{algo:sub2}), an edge $e$ is included to $H$ if $x_e \geq \kappa$. Also, we choose a previously detected triangle $\Delta_e$ uniformly at random among the set of triangles charged to $e$. By \Cref{lem:subroutine2}, with high probability, ${x_e}/{p}$ is $(1+\eps/30)$-approximation to $t_e$, and for any triangle $\Delta$ (charged to $e$), we have $\Pr(\Delta_e=\Delta)$ is $(1\pm \eps/15)\cdot 1/t_e$. So, instead of storing $H$ explicitly, we can initiate a weighted reservoir sampler WRS (to store one element) and we give $(\Delta_e,x_e/p)$ to WRS when we detect a heavy edge $e$. Let the final sample reported by WRS be $(w,\Delta)$. We can argue that triangle $\Delta$ satisfies the guarantee by \samplheavy as stated in \Cref{lem:heavy--al-one}. It is important to note that the set $H$ also plays a crucial role in \sampllight (see Line 3 of \Cref{alg:1passallight} and \Cref{rem:light-heavy}). That is, in \sampllight, we check for each $e \in \cF_1$, whether $\mathsf{flag}(e)=1$ and $e \notin H$. But, now we are not storing $H$ explicitly. However, we run \rouone and \routwo in parallel. An edge  $e$ is possibly detected in \routwo as a heavy edge when $e$ arrives for the second time and $\mathsf{flag}(e)$ is possibly set to $1$ in \rouone when $e$ comes for the second time. So, when $e$ is detected to be a heavy edge, we can set $\mathsf{flag}(e)$ to $0$. Thus, the triangle $\Delta$ reported by \sampllight satisfies the guarantee as stated in \Cref{lem:light}. 

Putting things together, we conclude that, in one pass, the space complexity of \tsamp is $\tOh(m/\sqrt{T})$ in the \al model, hence we are done with the proof of \Cref{theo:al1pass_proof}.

\section{Conclusion}
\label{sec:conclude}
\noindent
In this work, we have studied the problem of triangle sampling in three popular streaming models. Our main contribution in this work is designing triangle sampling algorithms in the \al model, where the counting algorithms can not be generalized to get sampling algorithms, and we designed almost optimal algorithms in these scenarios. We showed that in \ea/\va models, triangle counting algorithms can be generalized to design triangle sampling algorithms. 

The main open question left is to study this problem in the Random Order Model (\ra model), where the edges of the graph appear in random order. In this model, the current best space complexity of triangle counting is $\widetilde{\Oh}\left(m/\sqrt{T}\right)$ by McGregor and Vorotnikova~\cite{mcgregor2020triangle}. So the main open question is:
\begin{qs}
    Can one design a triangle sampling algorithm for \ra model matching the bounds of its counting variant?
\end{qs}

Another open question in this context is:
\begin{qs}
Can one design efficient sampling algorithms for other substructures in the streaming models?    
\end{qs}

%\newpage
\bibliographystyle{abbrv}
\bibliography{reference}
\appendix

\section{Remarks from \Cref{sec:intro}}\label{app:rem}

\begin{rem}{\bf (Sampling lower bounds from counting results).}\label{rem:lower}
The lower bounds stated in Theorem~\ref{theo:two}
follow from the  lower bounds for counting version of the problems~\cite{kallaugher2019complexity,braverman2013hard,bera2017towards}. All the triangle counting lower bounds are for the algorithms that distinguish between $0$ and $T$ triangles.  Let $G$ be the input graph for counting triangles. Let us now consider a new graph $G'=G \cup H$, where $H$ is a clique on $\Theta(T^{1/3})$ vertices. If we run $O(1)$ instances of our sampling algorithm over $G'$, then from the samples, we can distinguish whether $G$ contains $T$ triangles or $0$ triangles depending on whether  we get a triangle outside $H$ as a sample or not, respectively.
The formal statements of the lower bounds are presented in \Cref{sec:lb_app}.
\end{rem}

\begin{rem}{\bf (Regarding the threshold $T$).}
We assume that the parameter $T$ (a promised lower bound on $\size{t(G)}$) is at least polynomial in $\log n$ and $1/\eps$, that is, $T \geq \poly(\log n , 1/\varepsilon)$. Otherwise, the trivial streaming algorithm that stores all the edges matches the stated bounds in the above results. We assume that we know a lower bound on $T$. The promise $T$ on the number of structures  is a standard assumption for estimating substructures (such as triangles) in the world of graph streaming algorithms~\cite{DBLP:conf/focs/KallaugherKP18,DBLP:conf/icalp/KaneMSS12,DBLP:conf/pods/KallaugherMPV19, DBLP:conf/pods/McGregorVV16, bera2017towards}.~\footnote{Here we have cited a few. However, there are huge amount of relevant literature.}
 To make the exposition simple, While describing our algorithm, we assume that we know the exact values of $m$ and $\size{t(G)}=T$. We would like to note that the algorithm can be modified when we do not known $m$ and know a lower bound $T$ on $\size{t(G)}$ such that $T$ is at least a constant factor of $\Omega(\size{t(G)})$.  However, the algorithms can be suitably modified using standard techniques from graph streaming algorithms to work for unknown $m$ and a lower bound on $|t(G)|$. Our algorithms can adapt to the value of $m$ as stream progresses. The assumption on $T$ can be removed by running $\Oh(\log n)$ parallel instances of our algorithm with $T=T_1, \ldots, T=T_{\Oh(\log n)}$ such that $T_1$ is the promised lower bound on $\size{t(G)}$ and $T_{\Oh(\log n)}=\Theta(n^3)$. If we run the triangle counting algorithm in parallel, then at the end, we can know an approximate value $\hat{t}$ for $\Theta(\size{t(G)})$ such that $\frac{|t(G)|}{2}\leq \hat{t}\leq 2|t(G)|$. Then use the output of  $T=T_i$ such that $i$ is the highest index satisfying $T_i \geq \frac{1}{4}\hat{t}$.
 \end{rem}
\section{Triangle Sampling in \ea/\va Model}\label{sec:eaalgo}
\noindent
We discuss our triangle sampling algorithms for \ea and \va models. We start by first describing the $3$-pass algorithm.

\subsection{Multi-pass Sampling in \ea/\va Model 
%(Proof of Theorem~\ref{theo:al3pass})
}\label{sec:eamultipassalgo}

%\begin{theo}[Upper bound of Theorem~\ref{theo:one}(ii)]\label{theo:tsamp-arb-p-proof}\label{theo:tsamp-arb-p-proof_proof}

\begin{theo}\label{theo:tsamp-arb-p-proof}
\tsamp can be solved in $\tOh\left(m^{3/2}/T\right)$ space in \ea/\va model in $3$-passes.
\end{theo}

%\paragraph*{{\bf Multi-pass triangle sampling algorithm in the \ea/\va model:}}

Let us first discuss our intuition of the algorithm.
Let us consider an ordering $\prec$ on the  vertices of $G$ such that $u \prec v$ implies $\deg(u)\leq \deg(v)$. For a triangle $\{u,v,w\}$ with $u \prec v \prec w$, we charge the triangle to the edge $\{u,v\}$.

Let us discuss the following algorithm to sample a triangle uniformly at random. First let us take an edge $\{u,v\}$ uniformly at random from the graph. Without loss of generality, let us assume that $u \prec v$. If $\deg(u)\leq \sqrt{2m}$, then we take a neighbor $w$ of $u$ uniformly at random, then check whether $\{u,v,w\}$ is a triangle and $u \prec  v \prec w$. If yes, we report the triangle $\{u,v,w\}$ with probability $\deg(u)/\sqrt{2m}$. If $\deg(u)> \sqrt{2m}$, then choose a vertex $w$ from the graph such that the probability a particular vertex is chosen as $w$ is proportional to the degree of the vertex. Then check whether $\{u,v,w\}$ is a triangle and $u \prec  v \prec w$. If yes, we report the triangle $\{u,v,w\}$ as the output.

We can argue that the above procedure reports a triangle from the graph uniformly at random with probability $\Omega(T/m^{3/2})$. To boost the success probability, one can repeat the above procedure $\tOh{(m^{3/2}/T)}$ times.

Now we will describe our streaming implementation of this above intuition in \Cref{alg:3passea2}, and thereafter we analyse its correctness.

\begin{algorithm}
\caption{Triangle Sampling in $3$-pass \ea/\va Model}\label{alg:3passea2}
%\begin{algorithmic}
\SetAlgoLined

\textbf{Pass 1:} Sample two random edges $\{u, v\}$ and $\{x,y\}$ from the stream $\cM$. \label{algo:3passea_line1}

%\hspace{6pt}

%\label{algo:3passea_line2}

\textbf{Pass 2:} Find the degrees of $u$ and $v$, that is, $\mathrm{deg}(u)$ and $\deg(v)$, and the degrees of $x$ and $y$, that is, $\deg(x)$ and $\deg(y)$, respectively.\ \label{algo:3passea_line3}

%Find the degrees of $x$ and $y$, that is, $\deg(x)$ and $\deg(y)$.\ \label{algo:3passea_line4}

Find random neighbors $u'$ of $u$ and $v'$ of $v$, respectively.\ \label{algo:3passea_line5}

\textbf{Process:} Without loss of generality, let us assume $u \prec v$. Perform the following two operations parallely.\ \label{algo:3passea_line6}

\uIf{$\mathrm{deg}(u) \leq \sqrt{2m}$}{ \label{algo:3passea_line7}

Set $w \leftarrow u'$.\ \label{algo:3passea_line8}

Proceed to {\bf Pass 3} with probability $\deg(u)/\sqrt{2m}$.\ \label{algo:3passea_line9}

}

\uIf{$\deg(u) > \sqrt{2m}$}{ \label{algo:3passea_line10}

$w \leftarrow \mbox{Unif}\{x,y\}$.\ \label{algo:3passea_line11}

\uIf{$\deg(w) \leq \sqrt{2m}$}{ \label{algo:3passea_line12}

FAIL and ABORT the process.\ \label{algo:3passea_line13}

}

\uElse{ \label{algo:3passea_line14}

Proceed to {\bf Pass 3} with probability $\sqrt{2m}/\deg(w)$.\ \label{algo:3passea_line15}

}

}
%\label{algo:3passea_line16}
\textbf{Pass 3:} \uIf{$u,v$ and $w$ form a triangle and $u \prec v \prec w$}{ \label{algo:3passea_line17}

Output the triangle $\{u,v,w\}$.\ \label{algo:3passea_line18}

}

%Test if $u,v$ and $w$ forms a triangle and if $deg(u) < deg(w) < deg(v)$. If YES, output $(u,v,w)$.

%\end{algorithmic}

\end{algorithm}

Now let us prove Theorem~\ref{theo:tsamp-arb-p-proof}.

\begin{proof}[Proof of Theorem~\ref{theo:tsamp-arb-p-proof}]

We start with the space complexity.

%of our algorithm corresponding to Theorem~\ref{theo:tsamp-arb-p-proof}.

% Let us begin with the following lemma.

% \begin{lem}\label{lem:multipasseaalgo2}
% Algorithm~\ref{alg:3passea2} samples a triangle of $G$ uniformly. Moreover, it uses constant space during its execution.

% %with probability at least $1 - \frac{1}{m^c}$.
% \end{lem}

\paragraph*{Space Complexity:}
As seen in Algorithm~\ref{alg:3passea2}, all the passes require $\Oh(1)$ space. Since we are running $\widetilde{\Oh}(m^{3/2}/T)$ instances of Algorithm~\ref{alg:3passea2} in parallel, the total space used is $\widetilde{\Oh}(m^{3/2}/T)$.

%Let us first argue about the space complexity of Algorithm~\ref{alg:3passea2}. Note that in \textbf{Pass 1} and \textbf{Pass 2}, the algorithm samples two edges $(u,v)$ and $(x,y)$, and stores the degrees of the endpoints of the sampled edges, as well as a random neighbor of $u$ and $v$. Finally, in \textbf{Pass 3}, it checks two conditions, and decides about outputting the triangle $(u,v,w)$ based upon the outcomes of the conditions checked. Thus, Algorithm~\ref{alg:3passea2} uses constant space in a single execution.

%Since we are finally running $\tOh(\frac{m^{\frac{3}{2}}}{T})$ instances of Algorithm~\ref{alg:3passea2} in parallel, we conclude that $\tOh(\frac{m^{\frac{3}{2}}}{T})$ space is used in total.

\paragraph*{Proof of Correctness:}

Let us consider an arbitrary triangle $\Delta = \{u,v,w\}$ such that $u \prec v \prec w$. We argue that the probability that $\Delta$ is sampled by Algorithm~\ref{alg:3passea2} is $1/\sqrt{2}m^{3/2}$. Note that $\Delta$ has been sampled by Algorithm~\ref{alg:3passea2} only if the edge $\{u,v\}$ has been sampled in {\bf Pass 1}, as well as $w$ has been chosen in {\bf Pass 2}. Thus the probability that $\Delta$ has been sampled by Algorithm~\ref{alg:3passea2} can be written as follows:
\begin{equation}\label{eqn:probdelsampled}
\begin{aligned}
    \pr(\Delta \ \mbox{is been sampled}) = \pr(\{u,v\} \ \mbox{is  sampled in {\bf Pass 1}}) \times \ \pr(w \ \mbox{is chosen in {\bf Pass 2}}).
\end{aligned}    
\end{equation}

Now let us compute the second term in Equation~\ref{eqn:probdelsampled}. We consider the following two cases depending upon the degree of the vertex $u$.
\noindent
\paragraph*{\textbf{Case 1: $\deg(u) \leq \sqrt{2m}$:}} 
Note that the algorithm finds a random neighbor $u'$ of $u$ in {\bf Pass 2} in Line~\ref{algo:3passea_line5}. Then in the {\bf Process} step, the algorithm proceeds with probability $\deg(u)/\sqrt{2m}$, and sets $w$ as $u'$. Since $u'$ has been chosen randomly among all the neighbors of $u$, the vertex $w$ is sampled with probability $1/\deg(u)$. Hence, 
\begin{equation*}
\begin{aligned}
    \pr(\Delta \ \mbox{is been sampled}) = \frac{1}{m} \times \frac{\deg(u)}{\sqrt{2m}}\times \frac{1}{\deg(u)}=\frac{1}{\sqrt{2}m^{3/2}}.
\end{aligned}    
\end{equation*}
\noindent
\paragraph*{\textbf{Case 2: $\deg(u) > \sqrt{2m}$:}}
The vertex $w$ is chosen uniformly among the vertices $\{x,y\}$ with equal probability. In that case, the probability that $w$ is sampled is as follows:
$$\pr(w \in \{x,y\}) \times \frac{1}{2} \times \frac{\sqrt{2m}}{\deg(w)} = \frac{\deg(w)}{m} \times \frac{1}{2} \times \frac{\sqrt{2m}}{\deg(w)} = \frac{1}{\sqrt{2m}}.$$

Hence, 
\begin{equation*}
\begin{aligned}
    \pr(\Delta \ \mbox{is been sampled}) = \frac{1}{m} \times \frac{1}{\sqrt{2m}}=\frac{1}{\sqrt{2}m^{3/2}}.
\end{aligned}    
\end{equation*}
% \begin{description}
% \item[Case 1: $\deg(u) \leq \sqrt{2m}$:] 
% Note that the algorithm finds a random neighbor $u'$ of $u$ in Pass 2 in Line~\ref{algo:3passea_line4}. Then in the Process step, if $\deg(u) \leq \sqrt{2m}$, it proceeds with probability $\deg(u)/\sqrt{2m}$, and sets $w$ as $u'$. Since $u'$ has been chosen randomly among all the neighbors of $u$, the vertex $w$ is sampled with probability $1/\sqrt{2m}$.

% \item[Case 2: $\deg(u) > \sqrt{2m}$:]
% When $\deg(u) > \sqrt{2m}$, the vertex $w$ is chosen uniformly among the vertices $\{x,y\}$ with equal probability. In that case, the probability that $w$ is sampled is as follows:
% $$\pr(w \in \{x,y\}) \times \frac{1}{2} \times \frac{\sqrt{2m}}{\deg(w)} = \frac{\deg(w)}{m} \times \frac{1}{2} \times \frac{\sqrt{2m}}{\deg(w)} = \frac{1}{\sqrt{2m}}.$$

% \end{description}

% Since the edges $\{u,v\}$ and $\{x,y\}$ are sampled uniformly in Pass 1, following Equation~\ref{eqn:probdelsampled}, we can say that the probability that the triangle $\Delta$ has been sampled is:
% $(1/m) \times (1/\sqrt{2m}) = 1/\sqrt{2}m^{\frac{3}{2}}.$

Note that till now we have been arguing for a fixed triangle $\Delta$. Using the union bound over the $T$ triangles of $G$, we conclude that the probability that any triangle has been sampled by Algorithm~\ref{alg:3passea2} in a single execution is $T/\sqrt{2}m^{\frac{3}{2}}$.

Since we are running $t=\tOh(m^{\frac{3}{2}}/T)$ instances of Algorithm~\ref{alg:3passea2} in parallel, the probability that no triangle has been sampled is at most $(1- T/\sqrt{2}m^{\frac{3}{2}})^{t} \leq 1/\mathsf{poly}(n)$.

Now let us prove that Algorithm~\ref{alg:3passea2} has sampled a triangle uniformly among all the triangles of the graph $G$. So, let us now work on the conditional space that a triangle has been sampled by Algorithm~\ref{alg:3passea2}. Under this conditional space, the probability that the particular triangle $\Delta$ has been sampled is:
\begin{eqnarray*}
    \pr\left(\frac{\Delta \ \mbox{has been sampled}}{ \mbox{A triangle has been sampled}}\right) = \frac{\frac{1}{\sqrt{2}m^{\frac{3}{2}}}}{\frac{T}{\sqrt{2}m^{\frac{3}{2}}}} = \frac{1}{T}
\end{eqnarray*}

Thus the triangle $\Delta$ has been sampled uniformly by Algorithm~\ref{alg:3passea2}. This completes the proof of Theorem~\ref{theo:tsamp-arb-p-proof}.
\end{proof}

\subsection{$1$-Pass Triangle Sampling  in \ea/\va Model}\label{sec:ea1passalgo}

%\begin{theo}[Upper bound of Theorem~\ref{theo:one}(i)]\label{theo:tsamp-arb-1-proof}
\begin{theo}\label{theo:tsamp-arb-1-proof}
\tsamp can be solved by using $\tOh\left(\min\{m,m^{2}/T\}\right)$ space in \ea/\va model in $1$-pass.
\end{theo}

%\noindent{{\bf Single pass triangle sampling algorithm in the \ea/\va model:}}

Let us consider an ordering  $\pi:E \rightarrow  |E|$ over the edges of the graph such that $\pi$ is a bijection.~\footnote{In particular, $\pi$ refers to the order in which the edges of the graph will come in the streaming.} Let us discuss an algorithm to sample a triangle uniformly. Take two edges $\{u,v\}$ and $\{u',v'\}$ from the graph uniformly at random and with replacement. If the two edges have exactly one common vertex (say $u=u'$), then check whether $\{v,v'\}$ is an edge, i.e., $\{u,v,v'\}$ form a triangle. If yes and $\pi(\{v,v'\})>\pi(\{u,v\}), \pi(\{u,v'\})$, report the triangle as the output. We can argue that the above algorithm samples a triangle uniformly at random with probability $\Omega(T/m^2)$. We repeat the above procedure $\tOh{(m^2/T)}$ times to boost the success probability.

In order to implement the above intuition in the streaming, we use a reservoir sampler. The streaming implementation is described in \Cref{alg:1passea}, and thereafter we prove its correctness.

\begin{algorithm}
\caption{Triangle Sampling in $1$-pass \ea/\va Model}\label{alg:1passea}
%\begin{algorithmic}
\SetAlgoLined
%\SetKw{Continue}{continue}

Initiate a Reservoir sampler RS to store two edges $e_1, e_2$ uniformly at random with replacement from the stream.

\For{each edge $e$ in the streaming}
{
Give $e$ to as the input to the RS. 

After the execution of RS, 

$\{e_1',e_2'\} \leftarrow RS$. \

\uIf{$e \notin \{e_1',e_2'\}$}{

Check if $e$ forms an edge with the two edges $e_1',e_2'$ sampled by RS.

Let the respective triangle be $\Delta =\{e_1', e_2', e\}$. \

}

\uElse{
Delete any information (possible about some triangle) stored except the two edges (one of them is $e$ itself) sampled by RS. \

}

% if $e$ is not present as one of the two sample of $RS$, then check whether $e$ forms an edge with the two edges sampled by RS.

% Otherwise, delete any information (possible about some triangle) stored except the two edges (one of them is $e$ itself) sampled by RS.

}

\uIf{there is any triangle $\Delta$ stored}{

Output $\Delta$. \

}

\uElse{

Report {\sc Fail}. \

}

% If there is a triangle stored, then report it as the sample. 

% Otherwise, report {\sc Fail}.

\end{algorithm}

% \begin{algorithm}
% \caption{Triangle Sampling in $1$-pass \ea/\va Model}\label{alg:1passea}
% %\begin{algorithmic}
% \SetAlgoLined
% %\SetKw{Continue}{continue}

% $\mbox{{\sc Triangle-Sampled}} \leftarrow \emptyset, \mbox{{\sc Edge-Set}} \leftarrow \{e_1,e_2\}$.\ \label{algo:1passea_line1}

% %\State SWR: $(e_1,e_2) \leftarrow$ \mbox{current set of two sampled edges} along with a possible triangle $\Delta_1$ with $e_1$ and $e_2$ as two edges.

% %\label{algo:1passea_line2}

% %$\mbox{{\sc Edge-Set}} \leftarrow (e_1,e_2)$.\ \label{algo:1passea_line2}

% \For{each $e \in \cM \setminus \{e_1, e_2\}$}{ \label{algo:1passea_line3}

% %\State $(e_1,e_2) \leftarrow \mbox{{\sc Edge-Set}}$ %current edges in data structure

% $\{e_1',e_2'\} \leftarrow \cS(\mbox{{\sc Edge-Set}}, e)$.\ \label{algo:1passea_line4}

% \uIf{$\{e_1,e_2\} \neq \{e_1',e_2'\}$}{ \label{algo:1passea_line5}

% delete the triangle stored with the edges $\{e_1,e_2\}$.\ \label{algo:1passea_line6}

% }

% \uElseIf{$\{e_1,e_2\} = \{e_1',e_2'\}$}{ \label{algo:1passea_line7}

% \uIf{$\{e_1,e_2, e\}$ form a triangle}{ \label{algo:1passea_line8}

% $\mbox{{\sc Triangle-Sampled}} \gets \{e_1, e_2, e\}$.\ \label{algo:1passea_line9}
%     %store it with respect to the edges $(x,y)$.

% }

% }

% }

% \uIf{$\mbox{{\sc Triangle-Sampled}} \neq \emptyset$}{\label{algo:1passea_line10}

% Output $\mbox{{\sc Triangle-Sampled}}$.\ \label{algo:1passea_line11}

% }

% \end{algorithm}

Now we proceed to prove Theorem~\ref{theo:tsamp-arb-1-proof}.

\begin{proof}[Proof of Theorem~\ref{theo:tsamp-arb-1-proof}]

First we argue about space complexity.

\paragraph*{Space Complexity:}
%Let us first start by arguing the space complexity of the algorithm corresponding to Theorem~\ref{theo:tsamp-arb-1-proof}. From the description, it is clear that

The space complexity of \Cref{alg:1passea} is due to the reservoir sampler RS, which stores two edges, along with a (possible) triangle. This requires $\Oh(1)$ space in a single iteration. Since we run $\tOh(m^2/T)$ instances of Algorithm~\ref{alg:1passea} in parallel, the total space complexity is $\tOh(m^2/T)$.

%Algorithm~\ref{alg:1passea} uses only two sets as variables: $\mbox{{\sc Edge-Set}}$ and $\mbox{{\sc Triangle-Sampled}}$ which store $\Oh(1)$ edges in a single invocation of Algorithm~\ref{alg:1passea}. Since we run $\tOh(m^2/T)$ instances of Algorithm~\ref{alg:1passea} in parallel, the total space complexity is $\tOh(m^2/T)$.

%So Algorithm~\ref{alg:1passea} uses only constant space during its single execution.

\paragraph*{Proof of Correctness:}
%Let us now proceed with the correctness proof of Algorithm~\ref{alg:1passea}. 

% \begin{lem}\label{lem:eaalgo1}
% Algorithm~\ref{alg:1passea} samples a triangle of $G$ uniformly. Moreover, the amount of space used by Algorithm~\ref{alg:1passea} is constant.
% \end{lem}

% \begin{proof}
%First note that Algorithm~\ref{alg:1passea} uses only two sets as variables: $\mbox{{\sc Edge-Set}}$, $\mbox{{\sc Triangle-Sampled}}$ which stores constant number of edges in a single invocation of Algorithm~\ref{alg:1passea}. Thus, it is immediate that Algorithm~\ref{alg:1passea} uses constant space during its execution.

To prove that Algorithm~\ref{alg:1passea} samples a triangle uniformly, let us consider an arbitrary triangle $\Delta= \{e_1, e_2, e_3\}$ containing edges $e_1$, $e_2$ and $e_3$ such that $e_1$ and $e_2$ comes before $e_3$ in the streaming.

Let us consider the case when $\{e_1,e_2\}$ is sampled by the reservoir sample RS. Then observe that  the status if reservoir sampler is not changed on arrival of any edge $e$ in the stream, after the arrival of  $e_1$ and $e_2$. So, by the description of the algorithm, when $e_3$ arrives, the algorithm detects the triangle $\Delta=\{e_1,e_2,e_3\}$, Moreover, $\Delta$ will be the output reported by the algorithm.

Hence,

$$\pr \left(\Delta~\mbox{is reported}\right)=\pr(\mbox{RS samples $e_1$ and $e_2$})=\frac{1}{m^2}.$$

% Observe that $\Delta$ will be the output of the \Cref{alg:1passea} if $\{e_1,e_2\}$ is sampled by the reservoir sample RS. This is because, if 
% \begin{center}
% $E_1$:= Edges $e_1$ \mbox{and} $e_2$ stored in Line~\ref{algo:1passea_line1} have not been replaced\\
% by other edges till the algorithm terminates.

% $E_2$:= Algorithm~\ref{alg:1passea} outputs a sampled triangle.
% \end{center}
% First note that if the edges $e_1$ and $e_2$ have not been replaced during the execution of Algorithm~\ref{alg:1passea}, then when the edge $e_3$ appears in the stream, the triangle $\Delta$ will be sampled.
% Thus, the probability that $\Delta$ has been sampled is: $\pr(E_1) = 1/m^2$.

% % \begin{eqnarray*}
% % \pr(\Delta \ \mbox{has been sampled}) &=& \pr(E_1) = \frac{1}{m^2}
% % \end{eqnarray*}

% %$$E_1:= \Delta \ \mbox{has been sampled by Algorithm~\ref{alg:1passea}}$$
% % Let us now define another random variable $E_3$ as follows:
% % $$E_3:= \mbox{Algorithm~\ref{alg:1passea} outputs a sampled triangle}$$

Since there are $T$ triangles in the graph $G$ and the triangles reported by the algorithm are disjoint event, using the union bound, we can say that the probability that Algorithm~\ref{alg:1passea} outputs a sampled triangle is: $T/m^2$.

Since we are running $t= \tOh(m^2/T)$ instances of Algorithm~\ref{alg:1passea} in parallel, the probability that no triangle will be sampled among all the instances is bounded by $(1 - T/m^2)^{t} \leq 1/\mathsf{poly}(n)$. Thus, we conclude that with probability at least $1- 1/\mathsf{poly}(n)$, a triangle will be sampled.

Under the conditional space that a triangle has been sampled, the probability that a particular triangle $\Delta$ is sampled is as follows:
\begin{eqnarray*}
\pr(\Delta  \mbox{ is sampled}) = \frac{\pr(\Delta \ \mbox{is reported})}{\pr(\mbox{A triangle is reported})} 
= \frac{1/m^2}{T/m^2}= \frac{1}{T}.
\end{eqnarray*}
Thus, the triangle $\Delta$ has been sampled uniformly.
\end{proof}

\section{Lower Bound Results for Triangle Sampling}\label{sec:lb_app}

Here we discuss the lower bound results for \tsamp in different modes, as mentioned in the introduction. The lower bound results on \tsamp follow from the corresponding lower bound for triangle counting along with the discussion in \Cref{rem:lower} of \Cref{app:rem}.

\subsection{Lower bound for Triangle Sampling in \al Model}\label{sec:allb_app}
Let us begin by stating the multi-pass triangle counting lower bound in \al model.

\begin{lem}[Lower bound of Triangle Counting in Multi-pass, Theorem 5.2 of \cite{kallaugher2019complexity}]\label{theo:lower-tri-cnt}
For any $m,T \in \N$ such that $T \leq m^{3/2}$, there exist $m'= \Theta(m)$ and $T'=\Theta(T)$ such that any adjacency list streaming algorithm that
distinguishes between $m'$-edge graphs with $0$ and $T'$ triangles with probability at least $2/3$ in a constant number of passes requires $\Omega(m/T^{2/3})$ space.
\end{lem}
\begin{rem}\label{rem:nof}
    The lower bound stated in the above theorem is conditioned on the conjecture for the complexity of the \emph{Disjointness} problem in the \emph{number on forehead} model. Here, there are $k$ players each holding a vector in $\{0,1\}^n$ such that each player can see all the vectors except own vector. The objective of the players is to decide if there exists an index $i\in[n]$ such that the $i$-th coordinate of all the vectors are $1$. The conjecture states that the communication complexity of the problem is $\Omega(n)$ where as the best known lower bound is $\Omega(\sqrt{n})$~\cite{DBLP:journals/jacm/Sherstov14}.
\end{rem}

In the multi-pass setting, we have the following lower bound result for \tsamp  in the \al model (assuming that the conjecture states in \Cref{rem:nof} is true). 

\begin{theo}[Lower bound of \Cref{theo:two} (i)]
For any $m,T \in \N$ such that $T \leq m^{3/2}$, there exist $m'= \Theta(m)$ and $T'=\Theta(T)$ such that any adjacency list streaming algorithm that samples a triangle uniformly from any graph with $m'$ edges and $T'$ triangles with probability at least $2/3$ in a constant number of passes requires $\Omega(m/T^{2/3})$ space.
\end{theo}

%The above theorem follows from \Cref{theo:lower-tri-cnt} along with the discussion in \Cref{rem:lower}.

Now, we state the $1$-pass lower bound of triangle counting in \al model.

\begin{lem}[Lower bound for Triangle Counting in 1-pass, Theorem 5.1 of \cite{kallaugher2019complexity}]
For any $m, T \in \N$ such that $T \leq m$, there exist $m' = \Theta(m)$ and $T'=
\Theta(T)$ such that any \al streaming algorithm that distinguishes between $m'$-edge graphs with $0$ and $T'$ triangles in one pass with probability at least $2/3$ requires $\Omega(m/\sqrt{T})$ space.
\end{lem}

%Combining the above lemma with Observation~\ref{obs:samplingcountingreduction_app}, we have the following result.

In the one-pass setting, we have the following lower bound result for \tsamp  in the \al model. 

\begin{theo}[Lower bound of Theorem~\ref{theo:two} (ii)]
For any $m, T \in \N$ such that $T \leq m$, there exist $m' = \Theta(m)$ and $T'=
\Theta(T)$ such that any \al streaming algorithm that samples a triangle uniformly from any graph with $m'$ edges and $T'$ triangles in one pass with probability at least $2/3$ requires $\Omega(m/\sqrt{T})$ space.
\end{theo}

\subsection{Lower bounds for Triangle Sampling in \ea/\va Model}\label{sec:ealb_app}

Similar to the lower bounds for \al model, we will first state the lower bound of triangle counting in \ea/\va model below.

\begin{lem}[Lower bound for Triangle Counting in Multi-pass, Theorem 4.2 of \cite{bera2017towards}]
    For any $m, T \in \N$ such that $T = \Omega(m)$, there exist $m' = \Theta(m)$ and $T'=\Theta(T)$ such that any \ea (\va) streaming algorithm that distinguishes between $m'$-edge graphs with $0$ and $T'$ triangles in multi-pass with probability at least $2/3$ requires $\Omega(m^{3/2}/T)$ space.
\end{lem}

%Following Observation~\ref{obs:samplingcountingreduction_app}, we now have the following lower bound for triangle sampling.

In the \ea/\va model model, we have the following lower bound for triangle sampling in multi-pass.

\begin{theo}[Lower bound of \Cref{theo:tsamp-arb-p-proof}]
    For any $m, T \in \N$ such that $T = \Omega(m)$, there exist $m' = \Theta(m)$ and $T'=\Theta(T)$ such that any \ea (\va) streaming algorithm samples a triangle from any graph with $m'$ edges and $T'$ triangles in multi-pass with probability at least $2/3$ requires $\Omega(m^{3/2}/T)$ space.
\end{theo}

Similarly, let us first state the lower bound of triangle counting in $1$-pass \ea/\va model.

\begin{lem}[Lower bound for Triangle Counting in 1-pass, Theorem 1 of \cite{braverman2013hard}]
    For any $m, T \in \N$ such that $T = \Omega(m)$, there exist $m' = \Theta(m)$ and $T'=\Theta(T)$ such that any \ea (\va) streaming algorithm that distinguishes between $m'$-edge graphs with $0$ and $T'$ triangles in $1$-pass with probability at least $2/3$ requires $\Omega(m^2/T)$ space.
\end{lem}

%\comment{Sayantan: Check the reference and shall we write $\min\{m, m^2/T\}$?}

%Following Observation~\ref{obs:samplingcountingreduction_app}, and the above lemma, we have the following lower bound of triangle sampling in $1$-pass in \ea (\va) model.

So, in the \ea/\va model, we have the following lower bound of triangle sampling in $1$-pass.

\begin{theo}[Lower bound of \Cref{theo:tsamp-arb-1-proof}]
    For any $m, T \in \N$ such that $T = \Omega(m)$, there exist $m' = \Theta(m)$ and $T'=\Theta(T)$ such that any \ea (\va) streaming algorithm that samples a triangle from a graph with $m'$ edges and $T'$ triangles in $1$-pass with probability at least $2/3$ requires $\Omega(m^2/T)$ space.
\end{theo}

\section{Useful concentration bounds}\label{sec:conc_bounds} 

In our work, we use the following three concentration inequalities, see~\cite{dubhashi2009concentration}.

\begin{lem}[Chernoff-Hoeffding bound]
\label{lem:cher_bound1}
Let $X_1, \ldots, X_n$ be independent random variables such that $X_i \in [0,1]$. For $X=\sum\limits_{i=1}^n X_i$ and $\mu=\E[X]$, the following holds for all $0\leq \delta \leq 1$
$$ 
    \pr\left(\size{X-\mu} \geq \delta\mu\right) \leq 2\exp{\left(\frac{-\mu \delta^2}{3}\right)}.
$$

\end{lem}
\begin{lem}[Chernoff-Hoeffding bound]
\label{lem:cher_bound2}
Let $X_1, \ldots, X_n$ be independent random variables such that $X_i \in [0,1]$. For $X=\sum\limits_{i=1}^n X_i$ and $\mu_l \leq \E[X] \leq \mu_h$, the followings hold for any $\delta >0$.
\begin{itemize}
\item[(i)] $\pr \left( X \geq \mu_h + \delta \right) \leq \exp{\left(\frac{-2\delta^2}{n}\right)}$.
\item[(ii)] $\pr \left( X \leq \mu_l - \delta \right) \leq \exp{\left(\frac{-2\delta^2}{n}\right)}$.
\end{itemize}

\end{lem}

\begin{lem}[Hoeffding's Inequality] \label{lem:hoeffdingineq}

Let $X_1,\ldots,X_n$ be independent random variables such that $a_i \leq X_i \leq b_i$ and $X=\sum\limits_{i=1}^n X_i$. Then, for all $\delta >0$, we have
$$ \pr\left(\size{X-\E[X]} \geq \delta\right) \leq  2\exp\left(\frac{-2\delta^2} {\sum\limits_{i=1}^{n}(b_i-a_i)^2}\right).$$

\end{lem}

\end{document}